\pdfoutput=1

\documentclass[14pt, a4paper]{article}

\usepackage[utf8]{inputenc}
\usepackage[english]{babel}
\usepackage{amsmath}
\usepackage{jheppub}
\usepackage{amsfonts}
\usepackage{amssymb}
\usepackage{amsthm}
\usepackage{mathtools}
\usepackage{graphicx}
\usepackage{pdfpages}
\usepackage{subfig}
\usepackage{multicol}
\usepackage{pdfpages}
\usepackage{float}
\usepackage{tikz-feynman}
\usepackage{lscape}
\newtheorem{theorem}{Theorem}

\def\eq#1{Eq.~(\ref{#1})}
\newcommand{\secn}[1]{Section~\ref{#1}}
\newcommand{\be}{\begin{equation}}
\newcommand{\ee}{\end{equation}}
\newcommand{\beq}{\begin{eqnarray}}
\newcommand{\eeq}{\end{eqnarray}}
\newcommand{\bea}{\begin{eqnarray}}
\newcommand{\eea}{\end{eqnarray}}
\newcommand{\nn}{\nonumber}

\title{Integration-by-parts identities and differential equations for 
parametrised Feynman integrals} 

\author[a]{Daniele Artico,}
\author[b]{Lorenzo Magnea,}

\affiliation[a]{Institut f\"ur Physik, Humboldt-Universit\"at zu Berlin, \\
Newtonstra\ss e 15, D-12489, Berlin, Germany}
\affiliation[b]{Dipartimento di Fisica, Universit\`a di Torino, 
and INFN, Sezione di Torino, \\ Via Pietro Giuria 1, I-10125 Torino, Italy}

\emailAdd{daniele.artico@physik.hu-berlin.de}
\emailAdd{lorenzo.magnea@unito.it}

\abstract{Integration-by-parts (IBP) identities and differential equations are 
the primary modern tools for the evaluation of high-order Feynman integrals. 
They are commonly derived and implemented in the momentum-space 
representation. We provide a different viewpoint on these important tools 
by working in Feynman-parameter space, and using its projective geometry. 
Our work is based upon little-known results pre-dating the modern era of loop calculations~\cite{Regge:1968rhi,Ponzano:1969tk,Ponzano:1970ch,Regge:1972ns,
Barucchi:1973zm,Barucchi:1974bf}: we adapt and generalise these results, 
deriving a very general expression for sets of IBP identities in parameter 
space, associated with a generic Feynman diagram, and valid to any loop 
order, relying on the characterisation of Feynman-parameter integrands as 
projective forms. We validate our method by deriving and solving systems of 
differential equations for several simple diagrams at one and two loops, 
providing a unified perspective on a number of existing results.}

\begin{document}

\maketitle
  

\section{Introduction}
\label{Intro}

The calculation of high-order Feynman integrals is the cornerstone of the precision
physics program at present and future particle accelerators~\cite{Heinrich:2020ybq}.
The systematic development of modern methods to compute Feynman integrals, 
going beyond a direct evaluation of their parametric expression, began with the 
identification and explicit construction of Integration-by-Parts (IBP) identities in
dimensional regularisation, in Refs.~\cite{Tkachov:1981wb,Chetyrkin:1981qh},
and reached a further degree of sophistication with the development of the method
of differential equations~\cite{Kotikov:1990kg,Remiddi:1997ny,Gehrmann:1999as}.
These two sets of ideas can be combined into powerful algorithms~\cite{Laporta:2000dsw},
and the procedure further streamlined and optimised by the identification of the
linear functional spaces where (classes of) Feynman integrals live~\cite{Duhr:2011zq,
Duhr:2019wtr,Abreu:2022mfk}, and by taking maximal advantage of dimensional
regularisation~\cite{Henn:2013pwa}. The combined use of these tools has dramatically
extended the range of processes for which high-order calculations are available, 
and has broadened our understanding of the mathematics of Feynman integrals,
as reviewed for example in~\cite{Heinrich:2020ybq,Weinzierl:2022eaz}.

It is an interesting historical fact that the idea of studying and eventually computing 
Feynman integrals by means of IBPs and differential equations pre-dates all the 
developments just discussed, and was originally proposed not in the momentum 
representation, but in Feynman-parameter space. Of course, it is well-known that 
studies of Feynman diagrams flourished in the $S$-matrix era, as illustrated in 
the classic textbook~\cite{Eden:1966dnq}. In particular, the projective nature of
Feynman parameter integrands, and the importance of the monodromy properties 
of Feynman integrals under analytic continuation around their singularities, were
soon uncovered, and attracted the attention of mathematicians~\cite{Pham,Lascoux}
and physicists~\cite{Regge:1968rhi}. In this context, Tullio Regge and collaborators
published a series of papers~\cite{Ponzano:1969tk,Ponzano:1970ch,Regge:1972ns}  
studying the `monodromy ring' of interesting classes of Feynman graphs: first
the ones we would at present describe as `multi-loop sunrise' graphs in 
Ref.~\cite{Ponzano:1969tk}, then generic one-particle irreducible $n$-point 
one-loop graphs in Ref.~\cite{Ponzano:1970ch}, and finally the natural 
combination of these two classes, in which each propagator of the one-loop 
$n$-point diagram is replaced by a $k$-loop sunrise~\cite{Regge:1972ns}. 
All of these papers employ the parameter representation as a starting point, 
and make heavy use of the projective nature of the integrand.

At the time, these studies by Regge and collaborators did not immediately 
yield computational methods, but it is interesting to notice that, at least at the 
level of conjectures, several deep insights that have emerged in greater detail
in recent years were already present in the old literature. For example Regge,
in Ref.~\cite{Regge:1968rhi}, argues, on the basis of homology arguments, that 
all Feynman integrals must belong to a suitably generalised class of hypergeometric 
functions, an insight that was sharpened much more recently with the introduction 
of the Lee-Pomeransky representation~\cite{Lee:2013hzt} of Feynman integrals 
and the application of the GKZ theory of hypergeometric functions~\cite{GKZ1,
GKZ2,Klausen:2019hrg,Feng:2019bdx,Klausen:2021yrt,Klausen:2023gui}. 
Regge further argues that such functions obey sets of (possibly) high-order 
differential equations, which he describes as `a slight generalisation of the 
well-known Picard-Fuchs equations', also a recurrent theme~\cite{Lairez:2022zkj}.

While general algorithms were not developed at the time, two of Regge's collaborators,
Barucchi and Ponzano, were able to construct a concrete application of the general 
formalism for one-loop diagrams~\cite{Barucchi:1973zm,Barucchi:1974bf}. In those
papers, they show that for one-loop diagrams it is always possible to organise the 
relevant Feynman integrals into sets (that we would now call `families'), and find 
a system of linear homogeneous differential equation in the Mandelstam invariants
that closes on these sets, with the maximum required size of the system being 
$2^{n}-1$ for graphs with $n$ propagators\footnote{This counting has been reproduced
with modern (and more general) methods in \cite{Bitoun:2017nre,Bitoun:2018afx,
Mizera:2021icv}.}. These systems of differential equations were of interest 
to Barucchi and Ponzano because they effectively determine the singularity structure of 
the solutions, and thus the monodromy ring, in agreement with the general results of
Regge's earlier work. From a modern viewpoint, it is perhaps just as interesting
to use the system directly for the evaluation of the integrals, as done with
the usual momentum-space approach: this is the direction that we will pursue in
our exploratory study.

In the present paper, we start from the ideas of Refs.~\cite{Regge:1968rhi,Ponzano:1969tk,
Ponzano:1970ch,Regge:1972ns} and the concrete results of Barucchi and 
Ponzano~\cite{Barucchi:1973zm,Barucchi:1974bf} to propose a projective framework 
to derive IBP identities and systems of linear differential equations for Feynman integrals. 
In order to do so, we need to generalise the Barucchi-Ponzano results in several directions. 
First of all, those results predate the widespread use of dimensional regularisation, and 
do not in principle apply directly to infrared-divergent integrals. Fortunately, the projective 
framework naturally involves the (integer) powers of the propagators appearing in the 
diagram. These can be continued to complex values, providing a regularisation that is 
readily mapped to dimensional regularisation\footnote{Regge and collaborators also use 
this regularisation, having in mind mostly ultraviolet divergences, since the framework 
at the time was constructed for generic massive particles. They refer to the complex
values of the powers of the propagators as `Speer parameters', whereas we would 
now refer to this procedure as analytic regularisation.}. We are then able to show that 
the projective framework applies directly to IR divergent integrals, and we provide some 
examples. Next, we observe that the procedure to derive IBP identities in projective space 
generalises naturally to higher loops. Clearly, at two loops and beyond it would be of 
paramount interest to have a generalisation of the Barucchi-Ponzano theorem, guaranteeing 
the closure of a system of linear differential equations, providing an upper limit for its 
size, and giving a constructive procedure to build the system. This would require a much 
deeper understanding of the monodromy ring of higher-loop integrals. Lacking this 
knowledge (a gap which certainly points to promising avenues for future research), 
we can nonetheless apply the parameter-space IBP technique, and derive directly 
sets of differential equation on a case-by-case basis. Indeed, we show that the method 
can be successfully applied to two-loop integrals, and we provide examples, including 
the two-loop equal-mass sunrise, for which we recover the appropriate elliptic differential 
equation. Finally, we note that our application of the projective framework highlights 
the importance of boundary terms in IBP identities: contrary to the momentum-space 
approach in dimensional regularisation, boundary terms do not in general vanish in 
the projective framework: on the contrary, they may play an important role in linking 
complicated integrals to simpler ones, as we will see in concrete examples.

We note that the work presented in this paper is part of a recent revival of interest 
in the mathematical structure of Feynman integrals in parameter space, and presents
interesting potential connections to several current research topics in this context,
including intersection theory~\cite{Mastrolia:2018uzb,Frellesvig:2019kgj,Frellesvig:2019uqt,
Frellesvig:2020qot,Chestnov:2022xsy}, the concept of parametric annihilator~\cite{Bitoun:2017nre}, 
the use of syzygy relations in reduction algorithms~\cite{Agarwal:2020dye,Wang:2023nvh},
the study of generalised hypergeometric systems~\cite{Munch:2022ouq}, and the
reduction of tensor integrals in parameter space~\cite{Chen:2019mqc,Chen:2019fzm,
Chen:2020wsh}. More generally, for the first time in several decades we are witnessing
a rapid growth of our understanding of the mathematical properties of Feynman integrals, 
in particular with regards to analiticity and monodromy (see, for example,~\cite{Bourjaily:2020wvq,
Mizera:2021icv,Hannesdottir:2022bmo,Hannesdottir:2022xki}, and the lectures in 
Ref.~\cite{Mizera:2023tfe}), with potential applications to questions of phenomenological 
interest, such as the study of infrared singularities~\cite{Arkani-Hamed:2022cqe} and
the development of efficient methods of numerical integration~\cite{Borinsky:2020rqs,
Borinsky:2023jdv}.

The structure of our paper is the following. In \secn{Notat} we introduce our notation for 
the parameter representation of Feynman integrals and for Symanzik polynomials, briefly 
reviewing well-known material for the sake of completeness. In \secn{ProFo} we introduce
projective forms, and we use their differentiation and integration to lay the groundwork 
for the construction of IBP identities for generic projective integrals. In \secn{FeyPr} we
specialise our discussion to Feynman integrals, and give a general procedure to construct 
IBP identities in this case. In \secn{ExpEx} and in \secn{Twolex} we validate our results 
by discussing several concrete examples at one and two loops. Four Appendices 
give some further technical details on these examples.  Finally, in \secn{AssPer}
we present an assessment of our results and perspectives for future work. 


\section{Notations for parametrised Feynman integrals}
\label{Notat}

In this section we summarise some well-known basic properties of parametrised 
Feynman integrals, which will be useful in what follows. We adopt the notations 
of Refs.~\cite{Bogner:2007mn,Bogner:2010kv,Weinzierl:2022eaz}.

Consider a connected Feynman graph $G$, with $l$ loops, $n$ internal lines carrying 
momenta $q_i$ and masses $m_i$ ($i = 1, \ldots, n$), and $m$ external lines carrying 
momenta $p_j$ ($j = 1, \ldots, m$). At this stage we do not need to impose restrictions 
on external masses, so $p_j^2$ is unconstrained. On the other hand, momentum is 
conserved at all vertices of $G$, so one can parametrise the graph assigning $l$ 
independent loop momenta $k_r$ ($r = 1, \ldots, l$) to suitable edges of the graph. 
The line momenta are then given by
\beq 
  q_i \, = \, \sum_{r = 1}^l \alpha_{ir} k_r + \sum_{j = 1}^m \beta_{ij} p_j \, ,
\label{Incidence}
\eeq
where the elements of the {\it incidence matrices}, $\alpha_{ir}$ and $\beta_{ij}$, take values
in the set $\{-1,0,1\}$. Working in $d$ dimensions, with $d = 4 - 2 \epsilon$, and allowing for
the possibility of raising propagators to integer powers $\nu_i$ ($i = 1, \ldots, n$), one may 
associate to each graph $G$ a family of (scalar) Feynman integrals
\beq
  I_G \left( \nu_i, d \right) \, = \, (\mu^2)^{\nu - l d/2} \int \prod_{r = 1}^{l} \frac{d^d k_r}{{\rm i} 
  \pi^{d/2}} \, \prod_{i = 1}^n \frac{1}{\left( - q_i^2 + m_i^2 \right)^{\nu_i}} \, , 
\label{FeyInt}
\eeq
where we defined $\nu \equiv \sum_{i = 1}^n \nu_i$, and the integration must be performed 
by circling the poles in the complex plane of the loop energy variables according to Feynman's 
prescription. 

The integration over loop momenta in \eq{FeyInt} can be performed in full generality by means 
of the Feynman parameter technique, using the identity
\beq
\label{FPT}
  \prod_{i = 1}^n \frac{1}{\left (- q_i^2 + m_i^2 \right)^{\nu_i} } \, = \, 
  \frac{\Gamma(\nu)}{\prod_{j = 1}^n \Gamma(\nu_j)} \int_{z_j \geq 0} d^n z \,
  \delta \left(1 - \sum_{j = 1}^n z_j \right) \, \frac{ \prod_{j = 1}^n z_j^{\nu_j -1}}{\left( 
  \sum_{j = 1}^n z_j \left( -q_j^2 + m_j^2 \right) \right)^{\! \nu}} \, .
\eeq
By virtue of \eq{Incidence}, the sum in the denominator of the integrand in \eq{FPT} is a 
quadratic form in the loop momenta $k_r$, and can be written as 
\beq
\label{Mdef}
  \sum_{j = 1}^n z_j \left( - q_j^2 + m_j^2 \right) \, = \, - \sum_{r,s = 1}^l  M_{rs} \, k_r  \cdot k_s + 
  2 \sum_{r = 1}^l k_r \cdot Q_r + J \, ,
\eeq
where $M$ is an $l \times l$ matrix with dimensionless entries which are linear in the Feynman 
parameters $z_i$, $Q$ is an $l$-component vector whose entries are linear combinations of 
the external momenta $p_j$, and $J$ is a linear combination of the Mandelstam invariants
$p_i \cdot p_j$ and the squared masses $m_j^2$. Translational invariance of $d$-dimensional 
loop integrals allows to complete the square in \eq{Mdef}: the integral over loop momenta can 
then be performed, leading to 
\beq
\label{FP}
  I_G \left( \nu_i, d \right) \, = \, \frac{\Gamma(\nu - l d/2)}{\prod_{j = 1}^n \Gamma(\nu_j)}
  \int_{z_j \geq 0} d^n z \, \delta \left(1 - \sum_{j = 1}^n z_j \right) 
  \left( \prod_{j = 1}^n z_j^{\nu_j -1}\right) \,
  \frac{\mathcal{U}^{\, \nu - (l+1) d/2}}{\mathcal{F}^{\, \nu - l d/2}} \, ,
\eeq
where the functions
\beq
\label{Sym pol}
  \mathcal{U} \, = \, \mathcal{U} (z_i) \, = \, \det M \, , \qquad \quad 
  \mathcal{F} \, = \, \mathcal{F} \left( z_i, \frac{p_i \cdot p_j}{\mu^2}, \frac{m_i^2}{\mu^2}
  \right) \, = \, \det M \, \left( J + Q M^{-1} Q \right)/\mu^2 \, ,
\eeq
are called graph polynomials or Symanzik polynomials. References \cite{Bogner:2007mn,
Bogner:2010kv,Weinzierl:2022eaz} discuss in detail the properties of graph polynomials:
here we only note that both polynomials are homogeneous in the set of Feynman parameters,
$z_i$, with ${\cal U}$ being of degree $l$ and ${\cal F}$ of degree $l+1$; furthermore, both 
polynomials are linear in each Feynman parameter, with the possible exception of terms
proportional to squared masses in ${\cal F}$. These homogeneity properties set the stage
for employing the tools of projective geometry, as discussed below in \secn{ProFo}.

Remarkably, Symanzyk polynomials can be constructed directly from the connectivity
properties of the underlying Feynman graph. To do so, let us denote by $\mathcal{I}_G$
the set of the internal lines of $G$, each endowed with a Feynman parameter $z_i$.
A \textit{co-tree} $\mathcal{T}_G \subset \mathcal{I}_G$ is a set of internal lines of $G$ 
such that that the lines in its complement $\overline{\mathcal{T}}_G \subset \mathcal{I}_G$ 
form a spanning tree, {\it i.e.} a graph with no closed loops which contains all the vertices 
of $G$. The first Symanzik polynomial for the graph $G$ is then given by
\beq
  \label{U}
  \mathcal{U} \, = \, \sum_{\mathcal{T}_G} \prod_{i \in \mathcal{T}_G} z_i \, .
\eeq
Note that, in the case of an $l$-loop graph, one needs to omit precisely $l$ lines in order 
obtain a spanning tree: the polynomial ${\cal U}$ is therefore homogeneous of degree $l$, 
as announced. Similarly, we can consider subsets ${\cal C}_G \subset {\cal I}_G$ with the
property that, upon omitting the lines of ${\cal C}_G$ from $G$, the graph becomes a
disjoint union of two connected subgraphs. Clearly, each subset ${\cal C}_G$ defines
a {\it cut} of graph $G$, and contains $l+1$ lines. One may further associate with each
cut the invariant mass $s \left({\cal C}_G \right)$, obtained by squaring the sum of the 
momenta flowing in (or out) one of the two subgraphs -- by momentum conservation, 
it does not matter which subgraph we choose. The second Symanzik polynomial is 
then defined by 
\beq
\label{F}
  \mathcal{F} \, = \, \sum_{\mathcal{C}_G} \frac{\hat{s} \left( \mathcal{C}_G \right)}{\mu^2} \, 
  \prod_{i \in \mathcal{C}_G} z_i \, - \, \mathcal{U} \sum_{i \in \mathcal{I}_G} 
  \frac{m_i^2}{\mu^2} \, z_i \, .
\eeq
As expected, ${\cal F}$ is homogeneous of degree $l+1$ in the Feynman parameters.

To illustrate these rules, consider the one-loop box diagram depicted in Fig. 1a. As for any
one-loop diagram, it is immediate to see that the first Symanzik polynomial is simply
the sum of the Feynman parameters associated with the loop propagators. In this case
\beq
\label{Ubox}
  \mathcal{U} \, = \, z_1 + z_2 + z_3 + z_4 \, .
\eeq
The second Symanzik polynomial depends on kinematic data. If for example one picks 
massless on-shell external legs, all cuts involving two adjacent propagators vanish. One 
is then left with the Cutkosky cuts in the $s$ and $t$ channels. Defining $s = (p_1 + p_4)^2$ 
and $t = (p_1 + p_2)^2$ (with all momenta incoming), and assuming all internal masses 
to be the same, one finds
\beq
\label{FboxM}
  \mathcal{F} \, = \, \frac{s}{\mu^2} \, z_1 z_3 + \frac{t}{\mu^2} \, z_2 z_4 - 
  \frac{m^2}{\mu^2} \, \left( z_1 + z_2 + z_3 + z_4 \right)^2 \, .
\eeq
\begin{figure}[!h]
\vspace{-2mm}
\centering
\includegraphics[scale=.7]{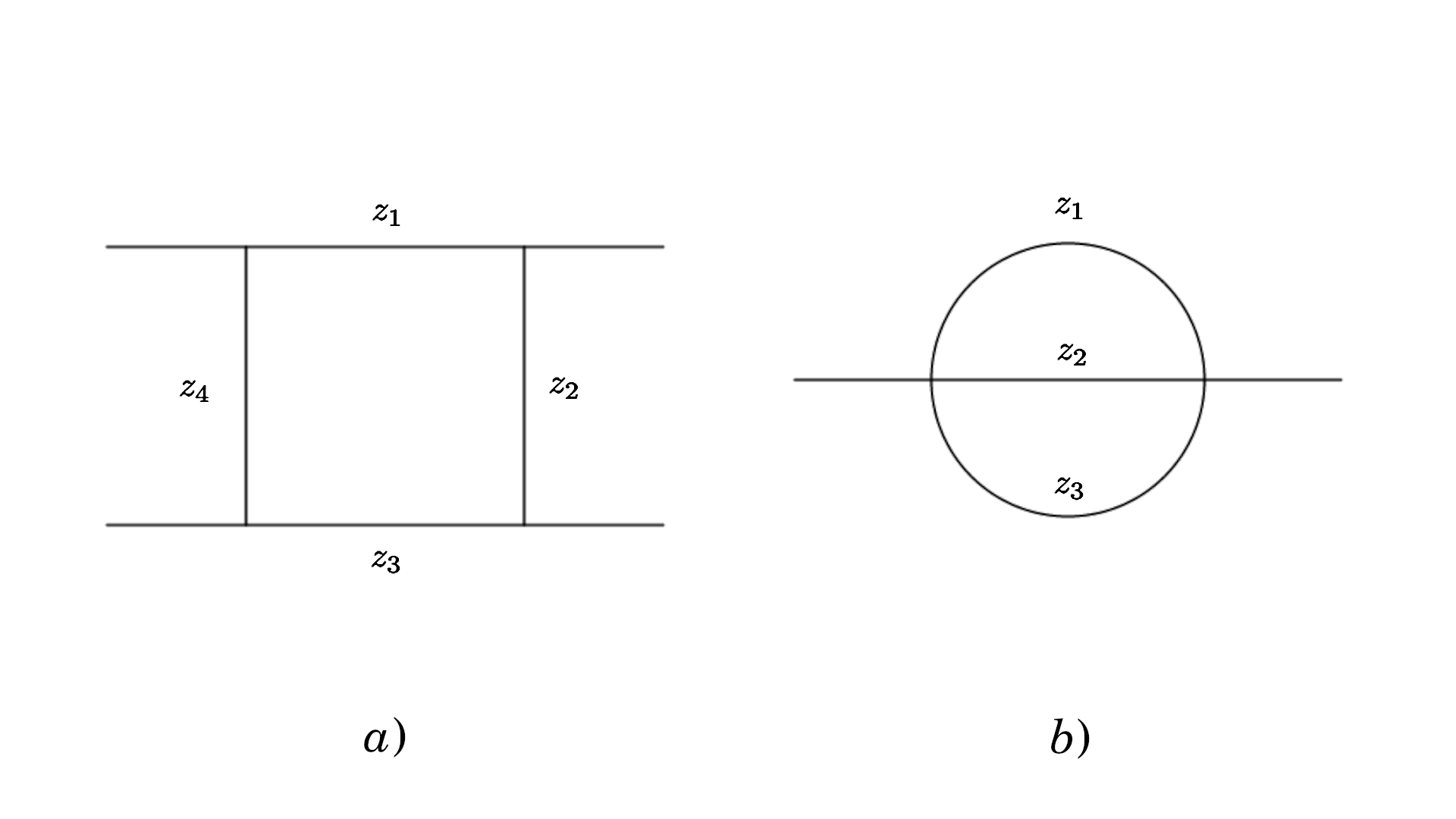}
\caption{a) One-loop box diagram b) Two-loop sunrise diagram}
\end{figure}
At two loops, one may consider the sunrise diagram in Fig. 1b. In this case, each internal 
line is a spanning tree (complementary to a co-tree). This implies that the first Symanzik 
polynomial is 
\beq
\label{U2}
  \mathcal{U} \, = \, z_1 z_2 + z_2 z_3 + z_1 z_3 \, .
\eeq
The cut-dependent part of the second graph polynomial is similarly straightforward, since 
only one cut exists. Taking equal internal masses, and denoting by $p^2$ the invariant 
mass of the incoming momentum, the second Symanzik polynomial is thus
\beq
\label{F2}
  \mathcal{F} \, = \, \frac{p^2}{\mu^2} \, z_1 z_2 z_3 - \frac{m^2}{\mu^2} \left( z_1 z_2 + 
  z_2 z_3 + z_1 z_3 \right) \left( z_1 + z_2 + z_3 \right) \, .
\eeq
In what follows, the crucial property of \eq{FP} is the {\it projective} nature of the integrand.
Indeed, one easily verifies that a change of variables of the form $z_i \to \lambda z_i$,
with $\lambda > 0$, leaves the integrand invariant, except for the argument of the $\delta$ 
function. Since a change of variables cannot affect the integral, we see that one should 
properly look at \eq{FP} as the integral of a projective form over the $(n-1)$-dimensional 
space $\mathbb{PR}^{n-1}$. This statement will be further substantiated in the next 
sections: in \secn{ProFo} we will show some technology concerning such integrals, 
which will then lead to a general integration-by-parts formula for Feynman parameter 
integrals in \secn{FeyPr}.


\section{Projective forms}
\label{ProFo}

In this section, we present a brief introduction to projective forms and to their integration 
and differentiation. Since the section is somewhat formal, it is useful to keep in mind from 
the beginning the announced correspondence between projective forms and Feynman-parameter 
integrands, which we will try to highlight with explicit examples.


\subsection{Preliminaries}
\label{Preli}

Let us begin by considering the Grassman algebra of exterior forms in the differentials $dz_i$, 
where $i \in D \equiv \left \{ 1, 2, ..., N \right\}$, for a positive integer $N$. Let $A$ 
be a subset of $D$, of cardinality $|A| = a$, and let $\omega_A$ be its ordered volume form
\beq
\label{wA}
  \omega_A \, = \, d z_{i_1} \wedge ... \wedge dz_{i_a} \, ,
\eeq
with $i_j \in A $, and $i_1 < i_2 < ... < i_a$: for example, if $A = \left\{ 1,2,3\right\}$, 
then $\omega_A = dz_{1} \wedge dz_2 \wedge dz_{3}$. The volume form $\omega_A$ can
be `integrated' by defining
\beq
\label{eta}
  \eta_A \, = \, \sum_{i \in A} \epsilon_{i, A - i} \, z_i  \, \omega_{A - i}
\eeq
where $A - i$ denotes the set $A$ with $i$ omitted, and we defined the signature factor 
$\epsilon_{k, B}$, for any $B \subseteq D$, and for any $k \notin B$, by means of
\beq
\label{epsilon}
  \epsilon_{k,B} \, = \, (-1)^{|B_k|} \, ,  \qquad   B_k \, = \, \left\{ i \in B, i < k \right\} \, ,
\eeq 
while $\epsilon_{k,B} = 0$ if $k \in B$. Using the properties of the boundary operator $d$,
one easily verifies that the differential of $\eta_A$ is proportional to $\omega_A$. Indeed 
\beq
\label{d eta}
  d \eta_A \, = \, a \, \omega_A \, .
\eeq
As an example, consider again $A = \left\{ 1, 2, 3 \right\}$: the form $\eta_A$ is 
then given by
\beq
\label{eta example}
  \eta_{ \{1,2,3\} } \, = \, z_1 \, dz_2 \wedge dz_3 - z_2 \, dz_1 \wedge dz_3 + z_3 \, dz_1 
  \wedge dz_2 \, ,
\eeq
and its differential in fact is equal to $3 \, dz_1\wedge dz_2 \wedge dz_3$. Consider next 
{\it affine} $q$-forms, defined by
\beq
\label{Affine}
  \psi_q \, = \, \sum_{|A| = q} R_A (z_i) \, \omega_A \, ,
\eeq
where $R_A$ is a homogeneous rational function\footnote{Note that this function may 
depend on external parameters as well, in our case representing kinematic invariants of
the diagram under consideration. Note also that, in order to make room for dimensional 
regularisation, we will slightly generalise this definition to include polynomial factors raised
to non-integer powers in both the numerator and the denominator of the functions $R_A$,
while preserving the homogeneity requirement.} of the variables $z_i$ with degree 
$- |A|  = - q$. The name affine form is a reference to their invariance under dilatations 
of all variables. \eq{Affine} is readily seen to imply that also the $(q+1)$-form $d \psi_q$ 
is affine. Anticipating \secn{TlSuI}, an example of an affine form with $q = 2$ and $N = 3$,
which is therefore the sum of three elements, is given by the integrand of the two-loop
sunrise diagram,
\beq
\label{Affine example}
  \psi_2 \left( \nu_i, \lambda, r \right) & = & \frac{z_1^{\nu_1 - 1} \, z_2^{\nu_2 - 1} \, 
  z_3^{\nu_3 - 1} \, (z_1 z_2 + z_2 z_3 + z_3 z_1)^\lambda}{\left( r \, z_1 z_2 z_3 - 
  (z_1 z_2 + z_2 z_3 + z_3 z_1) (z_1 + z_2 + z_3) \right))^{\frac{2 \lambda + \nu}{3}}}  \\ 
  && \hspace{3cm} \times \, \Big[ z_1 d z_2 \wedge d z_3 - z_2 d z_1 \wedge d z_3 + 
  z_3 d z_1 \wedge dz_2 \Big] \, , \nonumber
\eeq
where, as before, $\nu = \nu_1 + \nu_2 + \nu_3$. The parameter $\lambda$, which at
this stage is taken to be integer, will acquire a linear dependence on the dimensional 
regularisation parameter $\epsilon$ in the case of Feynman integrals, as discussed below.

An affine form is defined to be {\it projective} if it can be identically re-written as a linear
combination of the `integrated' forms $\eta_A$, defined in \eq{eta}. Then
\beq
\label{ProjFor}
  \psi_q \, = \, \sum_{|B| = q+1} T_B (z_i) \, \eta_B \, . 
\eeq
where the homogeneous functions $T_B (z_i)$ are obtained by suitably combining 
the functions $R_A (z_i)$ in \eq{Affine} with appropriate factors of $z_i$ arising from 
the definition of $\eta_A$ in \eq{eta}. As an example, the form $\psi_2$ in \eq{Affine example} 
is clearly projective, since the differentials reconstruct the form $\eta_A$ for $A = \{1,2,3\}$,
given in \eq{eta example}. Another example of a projective form that will appear in the 
following sections is the integrand of the one-loop massless box integral, which reads
\beq
\label{BoxExample}
  \psi_3 \left( \lambda, r \right) \, = \,
  \frac{(z_1 + z_2 + z_3 + z_4)^{\lambda}}{\left(r \, z_1 z_3 + z_2 z_4 
  \right)^{2 + \lambda/2}} \, \,
  \eta_{ \left\{ 1, 2, 3, 4 \right\}} \, ,
\eeq
where in the concrete application one will have $\lambda = 2 \epsilon$ and $r = t/s$.


\subsection{A useful theorem}
\label{Theo}

A useful result in what follows is the statement that the set of projective forms is closed
under differentiation. In other words
\begin{theorem}
The boundary of a projective form is itself projective.
\end{theorem}
\begin{proof}
Consider an operator $p$ trasforming an affine $q$-form into a projective (and therefore 
also affine) $(q - 1)$-form, according to 
\beq
\label{Defp}
  p \, : \,  \sum_{|A| = q} R_A (z_i) \, \omega_A \quad \rightarrow \quad
  \sum_{|A| = q} R_A (z_i) \, \eta_A \, .
\eeq
First, we note that the operator $p$ is nilpotent, \textit{i.e.} $p^2 = 0$. This can be 
easily shown for a single term in \eq{Defp}, $R_A (z_i) \, \omega_A$, and the generalization 
is then straightforward. In fact
\beq
\label{p^2}
  p^2 \Big( R_A (z_i) \, \omega_A \Big) & = & p \left( R_A (z_i) \sum_{i \in A} z_i \, 
  \epsilon_{i, A - i} \, \omega_{A - i} \right) \nonumber \\
  & = & R_A (z_i) \sum_{i > j, \left\{ i,j \right\} \in A} \! z_i z_j 
  \left( \epsilon_{i, A - i} \, \epsilon_{j, A - i - j} + 
  \epsilon_{j, A - j} \, \epsilon_{i, A - j - i} \right) \, = \, 0 \, .
\eeq
An example can serve the purpose of illustrating the cancellation in the last step:
\beq
\label{p^2ex}
  p^2 \Big( R_A (z_i) \, d z_1 \wedge d z_2 \wedge d z_3  \Big) & = & 
  R_A (z_i)  \, p \Big( z_1 \, d z_2 \wedge d z_3 - z_2 \, d z_1 \wedge d z_3 
  + z_3 \, d z_1 \wedge d z_2 \Big) \quad \\
  && \hspace{-3cm} = \, R_A (z_i) \, \big( z_1 z_2 d z_3 - z_1 z_3 d z_2 - z_2 z_1 d z_3 + 
  z_2 z_3 d z_1 + z_3 z_1 d z_2 - z_3 z_2 d z_1 \big) \, = \, 0 \, .
\nonumber
\eeq
The cancellation clearly works for any subset $A \subset D$, since there are always 
two terms in the sum that are proportional to $z_i z_j$, and they contribute with opposite 
sign. Considering for example $i < j$, when the factor of $z_j$ is generated by the first
action of $p$, the sign of the term is given by the position of the indices $i$ and $j$
in the ordered list of the elements of $A$. On the other hand, when the factor $z_i$ 
is generated by the first action of $p$, what is relevant is the position of $j$ in $A-i$.

The nilpotent operator $p$ can be combined with exterior differentiation to map affine
$q$-forms into affine $q$-forms. One can then show that
\beq
\label{dppd}
  d \circ p + p \circ d \, = \, 0 \, ,
\eeq
when acting on any affine $q$-form $\psi_q$. In order to prove \eq{dppd}, we note that
\beq
\label{pd}
  \big( p \circ d \big) \psi_q \, = \, \sum_{i \notin A} \, \frac{\partial R_A}{\partial z_i} \,
  (-1)^{|A_i|} \, \eta_{A \cup i} \, = \, \sum_{i \notin A} \sum_{j \in A \cup i} 
  \frac{\partial R_A}{\partial z_i} \, z_j \, (-1)^{|A_i|} \, (-1)^{|(A\cup i)_j|} \, 
  \omega_{A \cup i - j} \, , \quad
\eeq
while
\beq
\label{dp}
  \big( d \circ p \big) \psi_q \, = \, \sum_{j \in A} \,\, \sum_{i \notin A \, \vee \, i = j } 
  \left( \frac{\partial R_A}{\partial z_i} \, z_j + R_A (z_\ell) \, \delta_{i,j} \right) 
  (-1)^{|A_j|} \, d z_i \wedge \omega_{A - j} \, .
\eeq
By manipulating the indices and combining terms, the sum of \eq{pd} and \eq{dp} 
becomes
\beq
\label{dppd2}
  \Big( d \circ p + p \circ d \Big) \, \psi_q & = & \sum_{j \in A, \, i \notin A} 
  \frac{\partial R_A}{\partial z_i} \, z_j \left( (-1)^{|A_i|} \, (-1)^{| (A \cup i)_j |} + 
  (-1)^{| A_j |} \, (-1)^{| (A - j)_i |} \right) \omega_{A \cup i - j} \nonumber \\ 
  && \hspace{2cm} + \sum_{j \in A, \, i = j} R_A \, \delta_{i,j} \, \omega_A \,
  + \, \sum_{i \in D} \frac{\partial R_A}{\partial z_i} \, z_i \, \omega_A \, = \, 0 \, ,
\eeq
as desired. As an example of this last step, consider 
\beq
\label{exampleth}
  \psi_2 \, = \, \frac{z_1 + z_3}{(z_1 + z_2)^3} \, d z_1 \wedge d z_2 \, ,
\eeq
which implies
\beq
  d \psi_2 \, = \, \frac{1}{(z_1+z_2)^3} \, d z_1 \wedge d z_2 \wedge d z_3 
  \quad \longrightarrow \quad \big(p \circ d \big) \psi_2 \, = \, \frac{1}{(z_1+z_2)^3} 
  \, \eta_{\left\{ 1,2,3 \right\}} \, .
\eeq
On the other hand, one easily verifies that
\beq
  \big(d \circ p \big) \psi_2 & = &
  d \left( \frac{z_1 (z_1 + z_3)}{(z_1 + z_2)^3} \, d z_2 - 
  \frac{z_2 (z_1 + z_3)}{(z_1 + z_2)^3} \, d z_1 \right) \nn \\
  & = & - \frac{z_3}{(z_1 + z_2)^3} \, d z_1 \wedge d z_2 - \frac{z_1}{(z_1 + z_2)^3} \, 
  d z_2 \wedge d z_3 + \frac{z_2}{(z_1 + z_2)^3	} \,  d z_1 \wedge d z_3 \nn \\ 
  & = & - \frac{1}{(z_1 + z_2)^3} \, \eta_{\left\{ 1,2,3 \right\}} \, ,
\eeq
as desired. \eq{dppd}, just established, is actually sufficient to conclude the proof 
of the theorem. In fact, it can be shown~\cite{Regge:1968rhi} that all projective 
$q$-forms can be constructed by acting with the operator $p$ on $(q+1)$-forms:
in other words, they are $p$-exact, and any $\psi_q$ can be written as $\psi_q 
\, = \, p \big( \xi_{q+1} \big)$. An example can clarify this statement. Consider 
a generic affine two-form
\beq
\label{example for co-homology}
  \psi_2 \, = \, R_{12} \, dz_1 \wedge dz_2 + R_{13} \,dz_1 \wedge dz_3 + 
  R_{23} \, dz_2 \wedge dz_3 \, ,
\eeq
where $R_{ij}$ are homogeneous rational functions of degree $-2$ in the 
variables $z_1$, $z_2$ and $z_3$. By imposing that $p(\psi_2) = 0$ it is 
immediate to obtain 
\beq
\label{example for co-homology 2}
  \psi_2 \, = \, \frac{R_{12}}{z_3} z_3 dz_1\wedge dz_2 - \frac{R_{12}}{z_3} z_2 dz_1 
  \wedge dz_3 + \frac{R_{12}}{z_3} z_1 dz_2 \wedge dz_3 \, ,
\eeq
which is a projective form. The generalization to affine $n$-forms is straightforward, 
as the condition generates a system of linear equations that is enough to fix all the 
rational functions appearing in the form, but one (the overall coefficient function 
multiplying the projective volume form). Using this information on the {\it l.h.s} of 
\eq{dppd2}, one sees that the first term vanishes by the nilpotency of $p$; the second
term must then also vanish, which implies that $d \psi_q$ is itself $p$-exact 
and thus projective, as desired.
\end{proof}
\noindent In the context of Feynman parametric integration, the theorem is significant
for the following reason: given that Feynman integrals in the parameter representation
are integrals of projective forms on a simplex (as discussed below), applying the 
boundary operator $d$ on the integrand generates relations among forms with the 
same properties, \textit{i.e.} other Feynman integrands, or generalisations thereof.
These relations take the form of linear difference equations, which in turn can be 
used to build closed systems of differential equations to ultimately compute the
integrals, just as normally done in the momentum-space representation.


\section{Feynman integrals as projective forms}
\label{FeyPr}

This section presents the core results of our paper. We identify the integrands of 
Feynman integrals as projective forms of a specific kind, we examine their properties, 
and finally we use the fact that the differential of a projective form is still a projective 
form to write a set of generic relations among parametric integrands that include 
and generalise those appearing in Feynman integrals. These relations take the 
form of integration-by-parts (IBP) identities relating different (generalised) Feynman 
integrals, and can be used to build and simplify systems of differential equations 
in parameter space. 

To begin with, consider the projective form 
\beq
\label{ProjectiveForm}
  \alpha_{n-1} \, = \, \eta_{n-1} \,\, \frac{Q \big( \left\{ z_i \right\} 
  \big)}{D^P \big( \left\{ z_i \right\} \big)} \, , 
\eeq
where $\eta_{n-1}$ is the complete projective volume form of the projective space
$\mathbb{PC}^{n-1}$, while $Q(\left\{ z_i \right\})$ is a 
polynomial of degree $(l + 1) P - n$ and $D(\left\{ z_i \right\})$ a 
polynomial of degree $(l + 1)$. We recognise that the integrand of \eq{FP} is 
a specific instance of such a form, with the polynomial $D$ given by the second 
Symanzik polynomial of the graph, ${\cal F}$, and with $P = \nu - l d/2$. A first 
important property of projective forms such as \eq{ProjectiveForm}, and thus 
in particular of Feynman integrands, is the following
\begin{theorem}
Given two integration domains, $O, O' \in \mathbb{C}^n$,
if their image in $\mathbb{PC}^{n - 1}$ is the same simplex, 
then $\int_O \alpha_{n - 1} \, = \, \int_{O'} \alpha_{n-1}$.
\end{theorem}
\begin{proof}
This follows immediately from the fact that
\begin{itemize}
\item[i)] $\alpha_{n-1}$ is a closed form;
\item[ii)] $\eta_{n-1}$ is null on each surface defined by $z_i = 0$
\end{itemize}
Indeed, if we denote by $\Delta$ the subset of $\mathbb{C}^n$ given by the surface 
connecting points in the boundaries of $O$ and $O'$ that have a common image in 
the projective space, then $\int_{O + \Delta - O'} \alpha^{n - 1} = 0$ because of statement 
i), while $\int_{\Delta} \alpha^{n-1} = 0$ because of statement ii).
\end{proof}

\begin{figure}[!h]
\begin{center}
\includegraphics[scale=0.4]{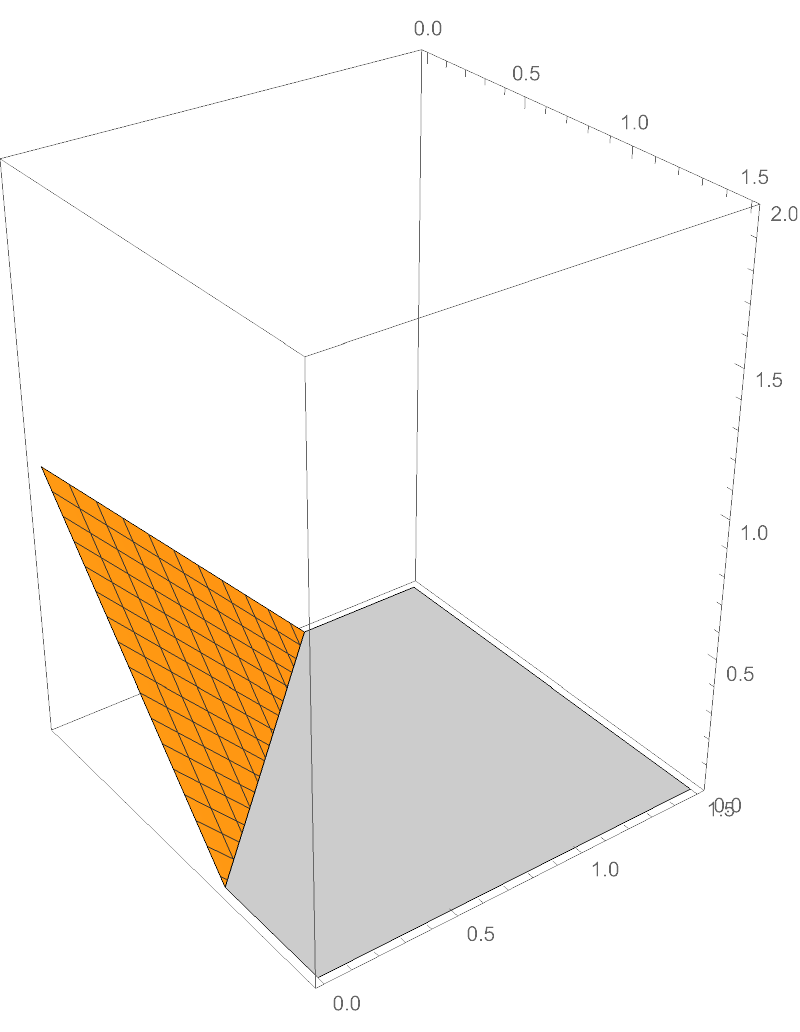}
\includegraphics[scale=0.4]{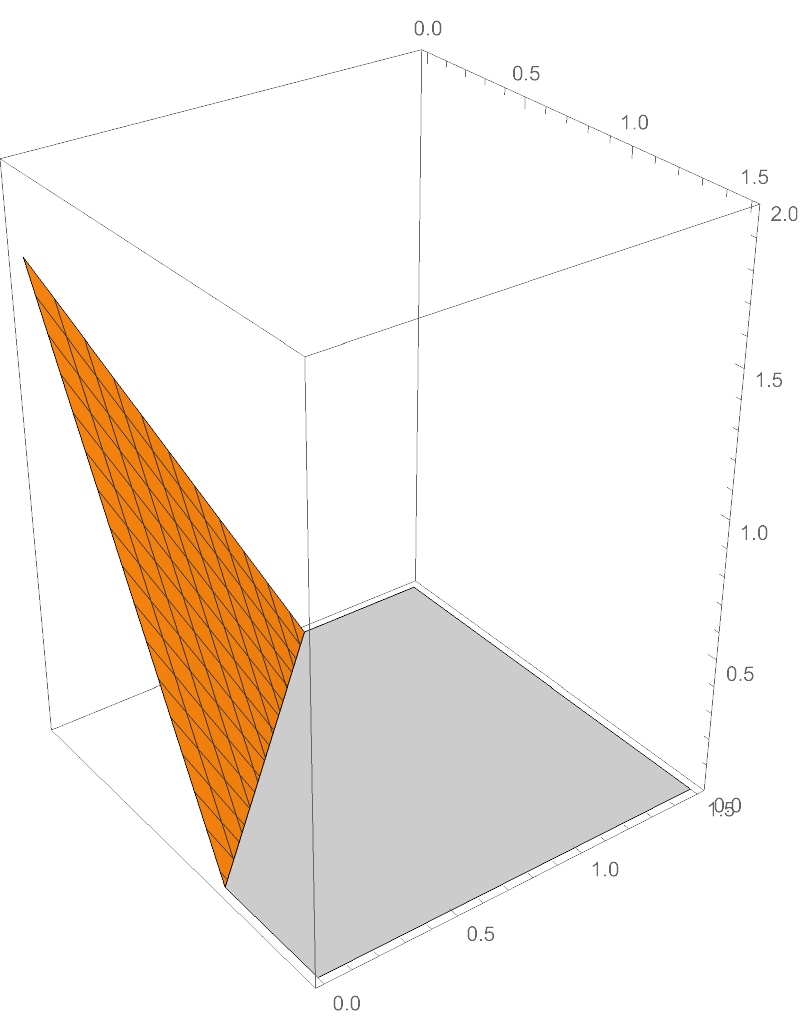}
\end{center}
\caption{Example of two domains in $\mathbb{R}^3$ that are equivalent under 
projective transformations.}
\end{figure}

We note that this theorem provides, in particular, a proof of the so-called Cheng-Wu 
theorem~\cite{Cheng:1987ga}. The latter, in its original form, states that in the argument 
of the delta function in the Feynman-parametrised expression, \eq{FP} it is possible 
to restrict the sum to an arbitrary subset of Feynman parameters $z_i$. In fact, consider 
the integration of \eq{ProjectiveForm} on the projection of the $n$-dimensional simplex, 
$S_{n-1} \equiv \left\{ z_i  \, | \, \sum_{i=1}^n z_i = 1 \right\}$. This choice of
integration domain is arbitrary (within the set of projectively equivalent domains), as 
was proven above. This means that, for example, the set defined by $\{z_i \, | \, 
\sum_{i = 1}^n z_i = 1\}$ and the set $\{z_i \, | \, z_n + t  \sum_{i = 1}^{n-1} z_i = t\}$ 
are equivalent for any positive value of $t$. In the limit $t \to \infty$ the integration 
domain becomes independent of $z_n$, and becomes a semi-infinite $(n-1)$-dimensional
surface based on the simplex $\{z_i \, | \, \sum_{i = 1}^{n-1} z_i = 1\}$. Figure 2 provides 
an example in three dimensions of the two mentioned surfaces.
Given these preliminary considerations, we can proceed using the conventional choice 
of $S_{n-1}$ as integration domain. In that case
\beq
\label{sumdz}
  d z_n \, = \, - \sum_{i = 1}^{n-1} d z_i \, ,
\eeq
so that
\beq
\label{projective F}
  \int_{S_{n - 1}} \, \eta_{n - 1} \, \frac{Q(z)}{D^P(z)} \, = \, \int_{z_i \geq 0}
  d z_1 \ldots d z_n \, \delta \! \left(1 - \sum_{i=1}^n z_i \right) \, \frac{Q(z)}{D^P(z)} \, ,
\eeq
where we use the shorthand notation $z$ for the set $\{z_i\}$. Any consistent choice
of the polynomials $Q(z)$ and $D(z)$, yielding a projective form, provides a natural
generalisation of \eq{FP}. We can now use the fact that the boundary of a projective
$(n-2)$-form is a projective $(n-1)$-form, to construct integration-by-parts identities
in Feynman parameter space in full generality. To this end, consider the projective 
$(n-2)$-form 
\beq
\label{w}
  \omega_{n-2} \, \equiv \, \sum_{i  = 1}^{n} (-1)^i \, \eta_{\{z\} - z_i} \,
  \frac{H_i (z)}{(P - 1) \, \big( D(z) \big)^{P - 1}} \, ,
\eeq
with any suitable choice of the polynomials $H_i (z)$ and $D(z)$. Technically, we need 
to restrict the choice of $D(z)$ so that the singularities of $\omega_{n-2}$ lie in a general 
position with respect to the simplex integration domain $S_{n-1}$, and in particular they 
do not touch the sub-simplexes forming the boundaries of $S_{n-1}$. This case was 
labeled as case (A) in~\cite{Regge:1968rhi}. When the singularities reach the integration 
domain, it is necessary to perform a blow-up of the singular points and treat the singular 
regions separately. Note that dimensionally regularized UV and IR divergent integrals 
can be treated without difficulties: these divergences are regulated by the parameter 
$\epsilon$, which features in the exponents of the parameters, while the restriction on 
$D(z)$ is related to external kinematics, which may force the denominator to vanish on
boundary simplexes.

Acting now with the boundary operator $d$ on the form $\omega_{n - 2}$ gives
\beq
\label{dw}
  d \omega_{n-2} \, = \, \frac{1}{(P-1) \, \big( D(z) \big)^{P - 1}} \, \, \eta_{\{z\}} \, 
  \sum_{i = 1}^n \frac{\partial H_i(z)}{\partial z_i} \, - \, \frac{\eta_{\{z\}}}{\big( D(z) \big)^P} \, 
  \sum_{i = 1}^n  H_i \, \frac{\partial D(z)}{\partial z_i} \, .
\eeq
This is the sought-for integration-by-parts identity: the integration of any projective form 
of the kind introduced in \eq{ProjectiveForm}, with the choice
\beq
\label{formQ}
  Q \big( z \big) \, = \, \sum_{i = 1}^n  H_i \, 
  \frac{\partial D(z)}{\partial z_i} \, ,
\eeq
can be reduced to the integration of forms with smaller values of $P$, modulo a possible
boundary term, which can be integrated via the Stokes theorem on sub-simplexes (this is 
the reason for the requirement that the singular surface of $\alpha_{n-1}$ should not 
intersect the boundary). It is important to stress that the possible presence of non-vanishing 
boundary terms represents a substantial difference with respect to the conventional 
IBP identities in momentum space: in Feynman parametrisation, boundary terms do 
not in general integrate to zero. In terms of Feynman diagrams, the integration over 
sub-simplexes represents the shrinking of a line of the diagram to a point. \eq{dw} will 
be the basis for all the applications in the following sections. We note that, when applied 
to Feynman integrals, \eq{dw} is valid for any number of loops and external legs, since 
the structure of the integrands in parameter space can always be written as was done 
in \eq{ProjectiveForm}. In order to explore its applications, we begin by specialising to 
one-loop graphs.


\subsection{One-loop parameter-based IBP}
\label{OlIBP}

Consider \eq{FP} for a one-loop diagram with $n$ internal propagators. In this 
case one can write
\beq
\label{1LFP}
  I_G (\nu_i, d) \, = \, \frac{\Gamma(\nu - d/2)}{\prod_{j = 1}^n \Gamma(\nu_j)}
  \int_{z_j \geq 0} \! d^n z \, \delta \bigg(1 - z_{n+1} \bigg) \, 
  \frac{\prod_{j = 1}^{n+1} z_j^{\nu_j -1}}{\bigg[
  \sum_{i=1}^{n+1} \sum_{j=1}^{i-1} s_{ij} z_i z_j \bigg]^{\nu - d/2}} \, , \quad
\eeq
where we introduced the matrix $s_{ij}$ $(i,j = 1, \ldots, n+1)$, defined by 
\beq
\label{sij}
  s_{ij} \, = \, \frac{(q_j - q_i)^2}{\mu^2} \quad (i,j = 1, \ldots, n) \, , \qquad \quad
  s_{i, n+1} \, = \, s_{n+1, i} \, \equiv \, - \frac{m_i^2}{\mu^2} \, ,
\eeq
as well as the auxiliary quantities
\beq
  z_{n+1} \, \equiv \, \sum_{i =1}^n z_i \, , \qquad \quad \nu_{n+1} \, \equiv \, \nu - d + 1 \, .
\eeq
The $s_{ij}$ represent then the invariant squared masses of the combinations of 
external momenta flowing in or out the diagram between line $i$ and line $j$.
We can now consider \eq{w}, and choose for $H_i$ simply the numerator of 
\eq{1LFP}. Furthermore, we can consider separately the $n$ forms obtained by
omitting from the projective volume the variable $z_i$, ($i = 1, \ldots, n$), in turn. 
This amounts to setting
\beq
\label{chooseHi}
  H_i \, = \, \delta_{ih} \left( \prod_{j = 1}^n z_j^{\nu_j -1} \right) 
  \left(\sum_{k =1}^n z_k \right)^{\nu - d} 
  \, = \, \, \delta_{ih} \, \prod_{j = 1}^{n+1} z_j^{\nu_j -1} \, ,
 \eeq
for some $h \in \left\{ 1,...,n\right\}$. Supposing $\nu_h > 1$, this leads to 
\beq
\label{1LIBP} 
  && d \left( (-1)^h  \, \eta_{\{z\} - z_h} \frac{ \prod_{j = 1}^{n+1} z_j^{\nu_j -1}}{\big( 
  \nu - (d + 1)/2 \big) \left( \sum_{i = 1}^{n+1} 
  \sum_{j=1}^{i-1} s_{ij} z_i z_j \right)^{\nu - (d+1)/2 }} \right) \, = 
  \nonumber \\
  && \hspace{5mm} = \, \frac{\eta_{\{z\}}}{ \big( \nu - (d+1)/2 \big) \left( \sum_{i = 1}^{n+1} 
  \sum_{j=1}^{i-1} s_{ij} z_i z_j  \right)^{\nu - (d+1)/2 }} \,
  \left[ (\nu_h -1) \frac{H_i}{z_h} + (\nu - d) \, \frac{H_i}{z_{n+1}} \right] 
  \nonumber \\
  && \hspace{5mm} - \, \frac{ \eta_{\{z\}} }{\left( \sum_{i = 1}^{n+1} 
  \sum_{j=1}^{i-1} s_{ij} z_i z_j \right)^{\nu - (d-1)/2 }} \, H_i \, \left( \sum_{k = 1}^{n+1} 
  \big( s_{k h} + s_{k, n+1} \big) z_k \right) \,.
\eeq
where $s_{n+1,n+1} = 0$. In order to clarify the structure of this expression, we introduce 
an index notation, following Ref.~\cite{Barucchi:1973zm}. We first define
\beq
\label{notforf1}
  f \big( \left\{ \nu_1, \ldots, \nu_{n+1} \right\} \big) \, \equiv \, 
  f \big( \left\{ {\cal R} \right\} \big) \, = \, \eta_{\{z\}} \, \frac{
  \prod_{j = 1}^{n+1} z_j^{\nu_j -1}}{\left( 
  \sum_{i = 1}^{n+1} \sum_{j=1}^{i-1} s_{ij} z_i z_j \right)^{\nu - d/2 }} \, .
\eeq
Then we write
\beq
\label{notforf2}
  f \big( \left\{ {\cal I} \right\}_{-1}, \left\{ {\cal J} \right\}_{0}, \left\{ {\cal K} \right\}_{1})
\eeq
to denote the same function as in \eq{notforf1}, where however the indices $\nu_i \in 
\left\{ {\cal I}, {\cal J}, {\cal K} \right\}$ have been respectively decreased 
by one, left untouched, and increased by one. Note that we consider sets such 
that $ \left\{ {\cal I} \right\} \cup \left\{ {\cal J} \right\} \cup \left\{ {\cal K} \right\} = 
\mathcal{R}$. Furthermore, the raising and lowering operations are defined in 
order to preserve the character of $f$ as a projective form, so the exponent of 
the denominator is re-determined after raising and lowering the indices.
With this notation, \eq{1LIBP} can be written as 
\beq
\label{1LIBP2}
  d \omega_{n-2}  + \sum_{k = 1}^{n+1} (s_{k h} + s_{k, n+1}) \,
  f \big( \left\{ \mathcal{R} - k \right\}_0,\left\{ k \right\}_{1} \big) & = & 
  \frac{\nu_h - 1}{\nu - (d+1)/2} \, f \big( \left\{ h \right\}_{-1}, 
  \left\{ \mathcal{R} - h \right\}_0 \big) \\
  & + &  \frac{\nu - d}{\nu - (d+1)/2} \, f \big( \left\{ n+1 \right\}_{-1}, 
  \left\{ \mathcal{R} - \left\{ n+1 \right\} \right\}_0 \big) \, . \nonumber
\eeq
This is the desired `integration by parts identity' at one loop, which at this stage 
is kept at integrand level to emphasise the fact that the boundary integral is not 
a priori vanishing. 
Notice that at the one-loop level one also has the constraint
\beq
\label{sum of f}
  \sum_{i = 1}^n f \big( \left\{ \mathcal{R} - i \right\}_0, \left\{ i \right\}_{1} \big) \, = \,  
  f \big( \left\{ \mathcal{R} - \left\{ n+1 \right\} \right\}_0, \left\{ n+1 \right\}_{1} \big) \, ,
\eeq
immediately following from the definition of $f$ and the one-loop Symanzik 
polynomial.

Ref.~\cite{Barucchi:1973zm} shows that, when the boundary term is zero, 
\eq{1LIBP2} and \eq{sum of f} allow for the systematic construction of a closed 
system of first-order differential equation. The proof proceeds by considering a 
set of parametric integrals containing all integrals obtained by raising an even 
number of parameter exponents by one unit (including the one of the first Symanzik 
polynomial). The derivatives of these integrals with respect to $s_{ij}$ are then included 
in a linear system of equations, as in \eq{1LIBP} and \eq{sum of f}, which is then 
solved in terms of the original set of integrals. This procedure is constructive and 
algorithmic, but one notices empirically that the number of integrals in the system 
is often higher than the number of actually independent master integrals, in concrete 
cases with specific mass assignments. In \secn{ExpEx}, we will use a similar 
construction, tryings to minimize the over-completeness of the resulting bases. 


\section{One-loop examples}
\label{ExpEx}


\subsection{One-loop massless box}
\label{OlMaB}

Let us consider the integral in \eq{1LFP} for the one-loop massless box integral, 
where $n = 4$, and for simplicity we set the renormalisation scale as $\mu^2 = 
(p_1+p_4)^2 \equiv s$, while $(p_1+p_2)^2 \equiv t$ (all momenta incoming). 
In particular, we focus on the simple case where all $\nu_i = 1$, and we define
\beq
\label{massless box}
  I_{\textrm{box}} \, = \, \Gamma (2 + \epsilon) \, \int_{S_{n-1}} \eta_{\{ z \}} \,
  \frac{ (z_1 + z_2 + z_3 + z_4)^{2 \epsilon}}{\left( r  z_1 z_3 + z_2 z_4 
  \right)^{2 + \epsilon}} \, \equiv \, \Gamma( 2 + \epsilon) \, I(1,1,1,1; 2 \epsilon) \, ,
\eeq
where, as in \eq{BoxExample}, we defined $r = t/s$, and the notation for four-point 
integrals is from now on $I(\nu_1, \nu_2, \nu_3, \nu_4; \nu_5)$. This 
notation is set up so that the arguments of the function correspond directly to 
the exponents of the propagators in the corresponding Feynman diagrams. 
Notice that in this framework, as is well known, the dimension of spacetime 
becomes simply a parameter related to the exponents of the first Symanzik 
polynomial, and dimensional shift identities are naturally encoded in the parameter 
based IBP equation, \eq{1LIBP2}. The matrix $s_{ij}$ for $I_{\rm box}$ reads 
\begin{equation}
s_{ij} = \left( \begin{array}
{c c c c c}

0 & 0 & r & 0 & 0 \\ 
 
0 & 0 & 0 & 1 & 0 \\ 

r & 0 & 0 & 0 & 0 \\ 

0 & 1 & 0 & 0 & 0 \\ 

0 & 0 & 0 & 0 & 0 \\ 

\end{array} \right)
\end{equation}
The construction in Ref.~\cite{Barucchi:1973zm} shows that the integrals that  
appear in the final differential equations system are $I(1,1,1,1;2\epsilon)$ and 
the ones obtained from it by raising by one the powers $\nu_i$ for an even 
number of parameters. Based on this result, we expect a basis set of integrals 
to be given by
\beq
\label{basisbox}
  \Big\{ I(1, 1, 1, 1; 2 \epsilon), I(2, 1, 2, 1; 2 \epsilon), I(1, 2, 1, 2; 2 \epsilon),
  I(2, 2, 2, 2; 2 \epsilon) \Big\} \,.
\eeq
In this simple case, we know that this basis is overcomplete: only three linearly 
independent master integrals are needed for the calculation of the one-loop 
massless box~\cite{Henn:2014qga}. Since however our goal here is just
to establish the viability of the method, we proceed with the ansatz in \eq{basisbox}.

To verify that this is indeed a basis and that we can close the system, consider first
the derivative of $I(1,1,1,1; 2 \epsilon)$ with respect to $r$,
\beq
\label{dzI1}
  \partial_r I (1, 1, 1, 1; 2 \epsilon) \, = \, - (2 + \epsilon) \, I(2, 1, 2, 1; 2 \epsilon) \, ,
\eeq
which indeed contains only integrals belonging to the desired set. On the other hand
\beq
\label{dzI2}
  \partial_r  I (2, 1, 2, 1; 2 \epsilon) \, = \, - ( 3 + \epsilon) \, I (3, 1, 3, 1; 2 \epsilon) \, .
\eeq
In order to proceed, it is necessary to express the integral $I(3, 1, 3, 1; 2 \epsilon)$ 
in terms of integrals belonging to the chosen set. \eq{1LIBP2} for $\nu_1 = 3$, 
$\nu_{3} = 2$, $\nu_2 = \nu_4 = 1$ and $h = 1$ becomes
\beq
\label{2010}
  r I(3, 1, 3, 1; 2 \epsilon) + \int d \omega_{n-2} \, = \, 
  \frac{2}{3 + \epsilon} I(2, 1, 2, 1; 2 \epsilon) + \frac{2 \epsilon}{3 + \epsilon} 
  I(3, 1, 2, 1; - 1 + 2 \epsilon) \, .                                                                                                                                                                                                                                                                                                                                                                                
\eeq
The boundary term is 
\beq
\label{vanbou}
  \frac{z_1^2 z_3 
  \left(z_1 + z_2 + z_3 + z_4 \right)^{2 \epsilon}}{(3 + \epsilon)
  (r z_1 z_3 + z_2 z_4)^{3 + \epsilon}}
  \big( z_2 d z_3 \wedge dz_4 - z_3 d z_2 \wedge d z_4 + z_4 
  d z_2 \wedge d z_3 \big) \Bigg|_{\partial S_{n-1}}
  \, = \, 0 \, ,
\eeq
which follows from the fact that the projective form $\eta^{234}$ vanishes on all boundary 
sub-simplexes, except the one defined by $z_1 = 0$, where however the integrand is zero. 
This property holds for all the identities used in this section, and the reasoning will
not be repeated. Consider now \eq{1LIBP2} for $\nu_1 = \nu_2 = \nu_3 = 2$, $\nu_4 = 1$ 
and $h = 2$, as well as the sum rule in \eq{sum of f} for $I(2, 1, 2, 1; - 1 + 2 \epsilon)$. 
One finds
\beq
\label{1110}
  I(2, 2, 2, 2; 2 \epsilon) & = & \frac{1}{3 + \epsilon} I(2, 1, 2, 1; 2 \epsilon) + 
  \frac{2 \epsilon}{3 + \epsilon} I(2, 2, 2, 1; - 1 + 2 \epsilon) \, , \nonumber \\
  I(2, 1, 2, 1; 2 \epsilon) & = & 2 I(3, 1, 2, 1; - 1 + 2 \epsilon) + 2 I(2, 2, 2, 1; - 1 + 2 \epsilon) \, ,
\eeq
where the symmetry of the integrand under the exchange of $(z_1, z_3)$ with $(z_2, z_4)$ 
has already been taken into account. This system and \eq{2010} allow to find a solution 
for $I(3, 1, 3, 1; 2 \epsilon)$, given by
\beq
\label{I2020}
  I (3, 1, 3, 1; 2 \epsilon) \, = \, \frac{1}{r} \Big[ I(2, 1, 2, 1; 2 \epsilon) -	 I(2, 2, 2, 2; 2 \epsilon) \Big] \, ,
\eeq
involving only integrals allowed in the system. Furthermore, one easily sees that the 
integral $I(2, 2, 2, 2; 2 \epsilon)$ is also involved in the equation 
\beq
\label{dI0101}
  \partial_r I(1, 2, 1, 2; 2 \epsilon) \, = \, - (3 +\epsilon) I(2, 2, 2, 2; 2 \epsilon) \, .
\eeq
The last derivative to be computed in terms of the chosen set of basis integrals
is $\partial_r I(2, 2, 2, 2; 2 \epsilon)$, which is proportional to $I(3, 2, 3, 2; 2 \epsilon)$. 
Using the same procedure adopted so far, it is possible to get a linear system of equation,
whose solution for the desired integral is 
\beq
\label{I2121}
  I(3, 2, 3, 2; 2 \epsilon) \, = \frac{I(2, 1, 2, 1; 2 \epsilon) - I(1, 2, 1, 2; 2 \epsilon) + 
  (3 + \epsilon) (1 + \epsilon + 3 r) I(2, 2, 2, 2; 2 \epsilon)}{(3 + \epsilon)
  (4 + \epsilon) \, r (1 + r)}
\eeq
The system of differential equations for our basis set of integral is now complete, 
and it reads
\beq
\label{system}
  \partial_r \textbf{b} \, \equiv \, \partial_r \left(
  \begin{array}{c}
  I(1, 1, 1, 1; 2 \epsilon) \\
  I(2, 1, 2, 1; 2 \epsilon) \\
  I(1, 2, 1, 2; 2 \epsilon) \\ 
  I(2, 2, 2, 2; 2 \epsilon)
  \end{array}
  \right) =  \left(
  \begin{array}{c c c c}
  0 & - (2 + \epsilon) & 0 & 0 \\
  0 & - \frac{3 + \epsilon}{r} & 0 & - \frac{3  + \epsilon}{r} \\
  0 & 0 & 0 & - (3 + \epsilon) \\
  0 & - \frac{1}{(3 + \epsilon) r ( 1 + r)} & \frac{1}{(3 + \epsilon) r (1 + r)} & - 
  \frac{1 + \epsilon + 3 r}{(3 + \epsilon) r (1 + r)} 
  \end{array}
  \right)
  \textbf{b} \, .
\eeq
Given \eq{system}, one can proceed using standard methods. In particular, \eq{system}
is not in canonical form~\cite{Henn:2013pwa}. Several techniques are available to 
solve this problem~\cite{Lee:2014ioa,Prausa:2017ltv,Gituliar:2017vzm,Lee:2020zfb}. 
Here, we simply follow the method of Magnus exponentiation \cite{Magnus:1954zz}: 
the necessary steps are presented in Appendix~\ref{App2}. Once the system is in 
canonical form, it can be solved iteratively as a power series in $\epsilon$ by standard 
methods\footnote{We notice that the final system of equations that we reach is not 
in `$d \log$' form, which is likely connected to the over-completeness of our basis. 
This is not a problem in this case, since the necessary iteration can be easily 
completed to the desired accuracy.}.

In the spirit of a proof-of-concept, we have not developed a systematic approach to 
the search for useful boundary conditions to determine the unique relevant solution 
of the system. In the case at hand, continuity in $r = -1$, uniform-weight arguments, 
and the known value of the residue of the double pole in $\epsilon$ can be used to 
recover the known solution. We find
\beq
  I_{\rm box} & = & \frac{k(\epsilon)}{r} \bigg[ \frac{1}{\epsilon^2} - \frac{\log r}{2 \epsilon} -
  \frac{\pi ^2}{4} + \epsilon \, \bigg( \frac{1}{2} \, {\rm Li}_3 (- r) - \frac{1}{2} \, {\rm Li}_2 (- r) 
  \log r + \frac{1}{12} \log ^3 r \nonumber \\
  && \hspace{2cm} - \, \frac{1}{4} \log(1 + r) \left(\log ^2 r + \pi ^2 \right) 
  + \frac{1}{4} \, \pi ^2 \log r + \frac{1}{2} \zeta(3) \bigg) \, + \, {\cal O} (\epsilon^2)
  \bigg] \, ,
\eeq
matching the result reported, for example, in Ref.~\cite{Henn:2014qga}. The one to one 
correspondence between the two results is found by setting the overall constant 
$k(\epsilon) \, = \, 4-\frac{\pi^2}{3}\epsilon^2 -\frac{40 \zeta(3)}{3}\epsilon^3$.


\subsection{One-loop massless pentagon}
\label{OlMaP}

We now turn to the natural next step, the one-loop massless pentagon. In dimensional 
regularisation, it is well-known that this integral can be expressed as a sum of one-loop 
boxes with one external massive leg (corresponding to the contraction to a point of one 
of the loop propagators), up to corrections vanishing in $d=4$. In this section, we will
recover this result, showing that, in this case, the method connects to the derivation
of the pentagon integral first reported in Ref.~\cite{Bern:1993kr}. From the point of view
of projective forms, the analysis of this case is interesting because it involves a non-vanishing 
boundary terms, contrary to what happened in \secn{OlMaB} and to the analysis of 
Ref.~\cite{Barucchi:1973zm}. 

Consider then \eq{1LIBP} for a five-parameter integral with $\nu_1 = \nu_2 = \nu_3 = 
\nu_4 = \nu_5 = 1$, and the exponent of the $\mathcal{U}$ polynomial equal to 
$2\epsilon$. Starting with the case $h = 1$, we obtain the equation
\beq
\label{I00000}
  \int_{S_{\left\lbrace 1,2,3,4,5 \right\rbrace}} \! d \omega_3 + s_{13} \, 
  I(1, 1, 2, 1, 1; 2 \epsilon) + s_{14} \, I(1, 1, 1, 2, 1; 2 \epsilon) 
  \, = \, \frac{2 \epsilon}{2 + \epsilon} I(1, 1, 1, 1, 1; - 1 + 2 \epsilon) \, , \qquad
\eeq
with 
\beq
  d \omega_3 \, = \, d \left[ - \, \eta_{\{2,3,4,5\}} \, 
  \frac{(z_1 + z_2 + z_3 + z_4 + z_5)^{2 \epsilon}}{(2 + \epsilon)
  \left( s_{13} z_1 z_3 + s_{14} z_1 z_4 + s_{24} z_2 z_4 + s_{25} z_2 z_5 + 
  s_{35} z_3 z_5 \right)^{2 + \epsilon}} \right] \, .
\label{dw3}
\eeq
Using Stokes theorem, and considering the only subset of the boundary of the 
five-dimensional simplex where $\eta_{\{2,3,4,5\}} \neq 0$, the boundary term 
of this equation becomes 
\beq
  \int_{S_{\{2,3,4,5\}}} \! \eta_{\{2,3,4,5\}} \, 
  \frac{(z_2 + z_3 + z_4 + z_5)^{2 \epsilon}}{\left(s_{24} z_2 z_4 + s_{25} z_2 z_5 + 
  s_{35} z_3 z_5 \right)^{2 + \epsilon}} \, = \,  
  I_{\rm box}^{(1)} (s_{25}) \, ,
\eeq
where $I_{\rm box}^{(1)}$ is a one-loop box integral with one massive external leg,
with a squared mass proportional to $s_{25}$. Effectively, the propagator with index 
$\nu_1$ has been contracted to a point. Note that, when applying Stokes theorem, 
the integration over boundary domains corresponds to the proper integration region, 
needed to obtain the lower-point Feynman integral, up to a sign arising from the 
orientation of the boundary. Reversing this orientation, when needed, produces 
a sign that, for example, cancels the minus sign in \eq{dw3}.

Considering, in a similar way, all the possible values for $h$, the following system 
of equations is obtained
\beq
\label{pentagon system}
  (2 + \epsilon) \!
  \left( \begin{array}{c c c c c}
  0 & 0 & s_{13} & s_{14} & 0 \\
  0 & 0 & 0 & s_{24} & s_{25} \\
  s_{13} & 0 & 0 & 0 & s_{35} \\
  s_{14} & s_{24} & 0 & 0 & 0 \\
  0 & s_{25} & s_{35} & 0 & 0
  \end{array}
  \right) \! \left( \begin{array}{c}
  I(21111;2\epsilon) \\
  I(12111;2\epsilon) \\
  I(11211;2\epsilon) \\
  I(11121;2\epsilon) \\
  I(11112;2\epsilon) \\
  \end{array}
  \right) \! + \!
  \left( \begin{array}{c}
  I_4^{(1)}(s_{25}) \\
  I_4^{(2)}(s_{13}) \\ 
  I_4^{(3)}(s_{24}) \\
  I_4^{(4)}(s_{35}) \\
  I_4^{(5)}(s_{14})
  \end{array}
  \right) \hspace{-0.2pt} = 
  2 \epsilon I(11111;-1+2\epsilon)
  \left( \begin{array}{c}
  1 \\
  1 \\
  1 \\
  1 \\
  1
  \end{array}
  \right) \! , \qquad
\eeq
where the integral $I(1,1,1,1,1; - 1 + 2 \epsilon)$ is proportional to the pentagon 
integral in $d = 6 - 2 \epsilon$. The solution to this system for the pentagon integral 
$I(1,1,1,1,1;1 + 2 \epsilon) = \sum_{i = 1}^5 I(\{i\}_1)$ is 
\beq
\label{pentagon}
  2 (2 + \epsilon) \,
  I(1, 1, 1, 1, 1;1 + 2 \epsilon) & = & 
  \Bigg\{\frac{s_{13} s_{24}-s_{13} s_{25}-s_{14} s_{25}+s_{14} s_{35}-
  s_{24} s_{35}}{s_{13} s_{14} s_{25}} \, I_{\rm box}^{(1)} \nonumber \\ 
  && - \, \frac{s_{13} s_{24}+s_{13} s_{25}-s_{14} s_{25}+s_{14} s_{35}-
  s_{24} s_{35}}{s_{13} s_{24} s_{25}} \, I_{\rm box}^{(2)} \nonumber \\ 
  && - \, \frac{s_{13} s_{24}-s_{13} s_{25}+s_{14}
   s_{25}-s_{14} s_{35}+s_{24} s_{35}}{s_{13} s_{24} s_{35}} \, I_{\rm box}^{(3)} \nonumber \\
  && + \, \frac{s_{13} s_{24}-s_{13} s_{25}+s_{14} s_{25}-s_{14} s_{35}-s_{24}
   s_{35}}{s_{14} s_{24} s_{35}} \, I_{\rm box}^{(4)} \nonumber \\ 
  && - \, \frac{s_{13} s_{24}-s_{13} s_{25}+s_{14} s_{25}+s_{14} s_{35}-
  s_{24} s_{35})}{s_{14} s_{25} s_{35}} \, I_{\rm box}^{(5)} \Bigg\} \nonumber \\ 
  && + \, 2 \epsilon \, I(1, 1, 1, 1, 1; - 1 + 2 \epsilon) \, ,
\eeq
recovering the result of Ref.~\cite{Bern:1993kr}. The correspondence between the 
coefficients reported here and those of Ref.~\cite{Bern:1993kr} can be derived using
the definition $c_i = \sum_{j = 1}^5 S_{ij}$ in their notation. A direct consequence of
\eq{pentagon} is the well-known theorem stating that the one-loop massless pentagon 
can be expressed as a sum of one-loop boxes with an external massive leg, up to 
$O(\epsilon)$ corrections. This last statement is due to the infrared and ultraviolet 
convergence of the $6-2\epsilon$ dimensional pentagon, which implies that the 
last line of \eq{pentagon} is $O(\epsilon)$.


\section{Two-loop examples}
\label{Twolex}

The first Symanzik polynomial for $l$-loop Feynman integrals, with $l > 1$, displays
a much more varied and intricate structure compared to the one-loop case, corresponding
to the factorially growing variety of graph topologies that can be constructed. Some 
classes of diagrams can still be described to all orders: a natural example is given by
the so-called $l$-loop sunrise graphs, depicted in Fig.~\ref{banana}, contributing to 
two-point functions and involving $(l+1)$ propagators. The monodromy ring for these 
graphs was identified in Ref.~\cite{Ponzano:1969tk}, but this result was not (at the 
time) translated into a systematic method to construct differential equations. The 
simplest non-trivial graph of this kind corresponds to $l = 2$, and we will discuss 
it below, in \secn{TlSuI}, in the case in which the masses associated with the three 
propagators are all equal. We will then consider the other non-trivial topology 
contributing to two-point functions at two loops, the five-edge diagram depicted 
in Fig.~\ref{fivedge}.


\subsection{Two-loop equal-mass sunrise integral}
\label{TlSuI}

Sunrise graphs at $l$ loops are characterised by the first Symanzik polynomial
\beq
  \mathcal{U}_l \, = \, \sum_{i = 1}^{l+1} z_1 \ldots \hat{z_i} \ldots z_{l+1} \, ,
\label{banfirstsym}
\eeq
where $\hat{z_i}$ is excluded from the product. Graphs of this class have generated 
a lot of interest in recent years. The two-loop sunrise graph with massive propagators 
is the simplest Feynman integral involving elliptic curves, and has been extensively 
studied both in the equal-mass case and with different internal 
masses~\cite{Broadhurst:1993mw,Muller-Stach:2011qkg,Adams:2013nia,
Adams:2014vja,Bloch:2013tra,Adams:2015gva,Adams:2015ydq,Bloch:2016izu,
Kalmykov:2016lxx,Broedel:2018qkq,Bogner:2019lfa}; furthermore, sunrise diagrams 
with massive propagators at higher loops provide early examples of integrals involving 
higher-dimensional varieties, notably Calabi-Yau manifolds~\cite{Bourjaily:2018yfy,
Broedel:2019kmn,Broedel:2021zij,Bonisch:2021yfw,Bourjaily:2022bwx}.

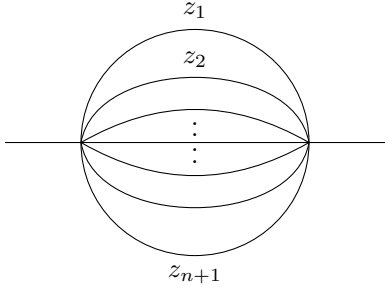
\begin{figure}[!h]
\centering
\begin{tikzpicture}
            \begin{feynman}
            \vertex (a1);
            \vertex[right=1cm of a1] (a2);
            \vertex[above right= 2.12 cm of a2](a3);
            \vertex[below right= 2.12 cm of a2](a4);
            \vertex[below = 1 cm of a3] (a31);
            \vertex[above = 1 cm of a4] (a32);
            \vertex[right=3cm of a2] (a5); 
            \vertex[right=1cm of a5] (a6); 
            \diagram* {
            (a1) -- (a2)
                -- [quarter left] (a3)
                -- [quarter left] (a5)
                -- [quarter left] (a4)
                -- [quarter left] (a2),
                (a5) -- (a6),
                (a2) -- (a5)
            };
            \end{feynman}
            \draw (a2) to[out=30,in=150] (a5);
            \draw (a2) to[out=-30,in=-150] (a5);
            \draw (a2) to[out=80,in=100] (a5);
            \draw (a2) to[out=-80,in=-100] (a5);
            \path
            (2.5,1.75) node {$z_1$}
            (2.5,1.1) node {$z_2$}
            (2.5,0.25) node {.}
            (2.5,0.10) node {.}
            (2.5,-0.25) node {.}
            (2.5,-0.10) node {.}
            (2.5,-1.75) node {$z_{n+1}$};
            \end{tikzpicture}
\caption{\label{banana} Sunrise diagram}
\end{figure}

In our present context, we would simply like to show how the projective framework 
that we are developing leads to the Picard-Fuchs differential equation obeyed by
the (equal-mass) two-loop sunrise integral \cite{Muller-Stach:2011qkg}. To this end, 
consider again equation \eq{Affine example}, which gives the relevant integral. In 
our present notation
\beq
\label{sunrise}
  I \big( \nu_1, \nu_2 , \nu_3 ; \lambda_4 \big) \, = \! \int_{S_{\{1,2,3\}}}
  \frac{\eta_3 \, z_1^{\nu_1 - 1} z_2^{\nu_2 - 1} z_3^{\nu_3 - 1} \, 
  \big( z_1 z_2 +  z_2 z_3 + z_3 z_1)^{\lambda_4}}{\Big[ r\,  z_1 z_2 z_3 - 
  (z_1 + z_2 + z_3) \big( z_1 z_2 + z_2 z_3 + z_3 z_1 \big) 
  \Big]^{\frac{2 \lambda_4 + \nu}{3}}} , \qquad
\eeq
where here $r = \frac{p^2}{m^2}$, and $p^\mu$ is the external momentum. For 
simplicity, we will work in $d = 2$, where the integral is finite both in the ultraviolet 
and in the infrared. In this case, the first Symanzik polynomial drops out, and the 
integrad is simply the inverse of the second graph polynomial. It is important to note 
that for this diagram both Symanzik polynomials vanish when approaching the boundary 
of the simplex $S_{\{1,2,3\}}$, at the points $z_i \to z_j \to 0$ and $z_k \to 1$. In 
principle, this configuration invalidates the application of Stokes theorem, as discussed 
in \secn{FeyPr}, and one needs to introduce a regularisation, for example by deforming 
the boundaries of the simplex near the corners~\cite{Muller-Stach:2011qkg,Adams:2014vja,
Bogner:2019lfa}. In the equal-mass case, the domain deformation can be avoided, 
since the corresponding corrections cancel: we will therefore proceed with the
general method, applying directly \eq{dw}. For an explicit discussion of the differences 
between the two cases, see Ref.~\cite{Weinzierl:2020xyy}.

Continuing with the strategy adopted at one loop, we use the numerator of
\eq{sunrise} (at this stage still for generic $d$) to define
\beq
  H (z) \, = \, z_1^{\nu_1 - 1} z_2^{\nu_2 - 1} z_3^{\nu_3 - 1} 
  \big(z_1 z_2 + z_2 z_3 + z_3 z_1 \big)^{\lambda_4} \, ,
\label{Hsunrise}\
\eeq 
which gives
\beq
  \frac{\partial H}{\partial z_h} \, = \, (\nu_h - 1) \, \frac{H}{z_h} + 
  \lambda_4 \, \frac{H}{\mathcal{U}_2} \, (z_j + z_k) \, , \hspace{1cm} h = 1,2,3 \, , 
  \quad j \neq k \neq h \, .
\label{dH/dx}
\eeq
Furthermore, denoting as before the square bracket in denominator of the 
integrand in \eq{sunrise} by $D(r)$, we find
\beq
  \frac{\partial D(r)}{\partial z_h} \, = \, - 2 \mathcal{U}_2 + (r - 1) z_j z_k  - z_j^2 - z_k^2 
  \, , \hspace{1cm} h = 1,2,3 \, , \quad j \neq k \neq h \, .
\label{dD/dx}
\eeq
Inserting \eq{dH/dx} and  \eq{dD/dx} into \eq{dw}, and picking the appropriate value of 
$P$ to ensure projective invariance, we arrive at the IBP equations
\beq
\label{dwSUN}
  d \omega_2 & = & \frac{3}{2 \lambda_4 + \nu - 1} \,
  \frac{\eta_3}{\big[ D(r) \big]^{\frac{2 \lambda_4 + \nu - 1}{3}}} 
  \left[ (\nu_h - 1) \, \frac{H}{z_h} + \lambda_4 \, \frac{H}{\mathcal{U}} \, (z_j+z_k) \right] 
  + \nonumber \\ 
  & & \hspace{5mm} - \, \frac{\eta_3}{\big[ D(r) \big]^{\frac{2 \lambda_4 + \nu + 2}{3}}}
  \, \big[ - 2 \mathcal{U} + (z - 1) z_j z_k  - z_j^2 - z_k^2 \, \big] \, H 
  \nonumber \\ 
  & = & \frac{3}{2 \lambda_4 + \nu - 1} \Big[ ( \nu_h - 1) \, f \big( \{h\}_{-1} \big) +
  \lambda_4 \, f \big( \{4\}_{-1}, \{j\}_1\big) + \lambda_4 \, f \big( \{4\}_{-1}, \{k\}_1 \big) 
  \Big] + \nonumber \\ 
  & & \hspace{5mm} - \, \Big[ - 2 f \big( \{4\}_1 \big) + (z - 1) f \big( \{j,k\}_1 \big) -
  f \big( \{j\}_2 \big) - f \big( \{k\}_2 \big) \Big] \, ,
\eeq
where in the second step we used the notation for raising and lowering operators 
in the function $f$ as discussed in \secn{OlIBP}. The functions $f$ are also related 
by the identity
\beq
  f \big( \{1,2\}_1 \big) + f \big( \{2,3\}_1 \big) + f \big( \{3,1\}_1 \big) \, =  \, 
  f \big( \{4\}_1 \big) \, .
\label{sunidf}
\eeq
Using the sum rule in \eq{sunidf}, and \eq{dwSUN}, we can build a linear system 
of equations involving the integrals $I(0,0,0,3\epsilon)$, $I(1,0,0,1+3\epsilon)$, 
and a non-vanishing boundary contribution $B$, arising from the IBP relation for 
$I(1,1,0,1+3\epsilon)$ when taking $h = 3$ (at this point, it should be clear that 
boundary terms only survive when $\nu_h = 1$). The linear system is presented in
Appendix~\ref{App4}, and the boundary term contributes to the equation
\beq
\label{J110}
  B & = & \int d \omega_1 \\
  & = & \frac{1 + 3 \epsilon}{1 + \epsilon} \, I(3, 2, 1; 3 \epsilon) + (1 - z) 
  I(3, 3, 1;1 + 3 \epsilon) + 2 I(4, 2, 1;1 + 3 \epsilon) + 2 I(2, 2, 1; 2 + 3 \epsilon)  
  \, , \nonumber
\eeq
where 
\beq
  \int d \omega_1 \, = \, \frac{1}{2 (1 + \epsilon)} \int_{S_{\{1,2\}}} \! \eta_{\{1,2\}} \,
  \frac{(z_1 z_2)^{\epsilon}}{\big[ - (z_1 + z_2) \big]^{2 + 2\epsilon}} \, = \, 
  \frac{(-1)^{2 \epsilon}}{2 + 2 \epsilon} \, 
  \frac{\Gamma^2(1 + \epsilon)}{\Gamma(2 + 2 \epsilon)} \, .
\label{bouter}
\eeq 
Note that the minus sign in the denominator and the factor of $(-1)^{2 \epsilon}$ 
come from the convention of including the masses with a minus sign in the second 
Symanzik polynomial. As stated above, we now set $d = 2$, so that the boundary 
term simply becomes $B = \frac{1}{2}$.

The linear system given in Appendix~\ref{App4} is sufficient to yield the following
non-homogeneous differential equations, involving two master integrals (the third
master integral appears here as the non-vanishing boundary term):
\beq\label{diffSUN}
\begin{cases}
 r \, \frac{d}{dr} I (1, 1, 1; 0) \, = \, I (1, 1, 1; 0) + 3 I (2, 1, 1; 1) \, , \\
 r (r - 1)(r - 9) \, \frac{d}{dr} I(2, 1, 1; 1) \, = \, (3 - r) I(1, 1, 1; 0) + 
 \left(9 - r^2 \right) I(2, 1, 1;1) + 2 r \, ,
\end{cases}.
\eeq
We note that the differential equation system in \eq{diffSUN} is the same 
reported in \cite{Laporta:2004rb}, up to a a different normalisation  of the 
non-homogeneous term, which is solely due to our different normalisation 
of Feynman integrals. This system can be transformed into a single second-order 
differential equation of Picard-Fuchs type by using the {\tt OreSys} package 
for {\tt Mathematica}: the result is
\beq
\label{elliptic}
  && \frac{r}{3} \, \frac{d^2}{d r^2} I(1, 1, 1; 0) \, + \, 
  \left( \frac{1}{3} + \frac{3}{r - 9} + \frac{1}{3 (r - 1)} \right) \frac{d}{d r} I(1, 1, 1; 0) 
  \nonumber \\
  && \hspace{2.6cm} \, - \, \left( \frac{1}{4 (r - 9)} + \frac{1}{12 (r - 1)} \right) I(1, 1, 1; 0) \, = \, 
  \frac{2}{(r - 1)(r - 9)} \, ,
\eeq
corresponding to the elliptic second order differential equation discussed in 
\cite{Broadhurst:1993mw,Laporta:2004rb}, up to our different normalisation. 


\subsection{Two-loop five-edge diagram}
\label{TlFiD}

As a last example, we consider the two-loop, five-edge diagram represented in 
Fig.~\ref{fivedge}, with all internal edges taken to be massless. In this way, the 
only kinematic parameter is the squared momentum $p^2$ carried by the external 
legs. This diagram has been extensively studied, starting with the seminal discussion
in Ref. \cite{Chetyrkin:1981qh}. In this section, the result of \cite{Chetyrkin:1981qh} 
is re-derived by using the parameter-space method presented in this article. 
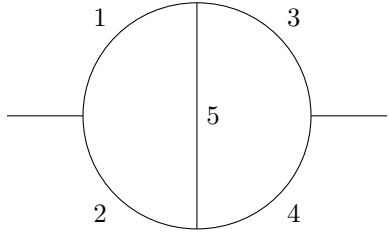
\begin{figure}[!h]
\centering
\begin{tikzpicture}
            \begin{feynman}
            \vertex (a1);
            \vertex[right=1cm of a1] (a2);
            \vertex[above right= 2.12 cm of a2](a3);
            \vertex[below right= 2.12 cm of a2](a4);
            \vertex[below = 1 cm of a3] (a31);
            \vertex[above = 1 cm of a4] (a32);
            \vertex[right=3cm of a2] (a5); 
            \vertex[right=1cm of a5] (a6); 
            \diagram* {
            (a1) -- (a2)
                -- [quarter left, edge label = 1] (a3)
                -- [quarter left,edge label = 3] (a5)
                -- [quarter left,edge label = 4] (a4)
                -- [quarter left,edge label = 2] (a2),
                (a5) -- (a6),
                (a3) -- [edge label = 5] (a4)        
            };
            \end{feynman}
            \end{tikzpicture}
\caption{\label{fivedge} Two-loop five-edge diagram}
\end{figure}

The graph polynomials for this diagram are given by
\beq
\label{U mezza mela}
  \mathcal{U} & = & (z_1 + z_2)(z_3 + z_4) + 
  z_5  \, \sum_{i = 1}^4 z_i \, , \nonumber \\
  \mathcal{F} & = & p^2 \big( z_1 z_2 z_3 + z_1 z_2 z_4 
  + z_1 z_3 z_4 + z_2 z_3 z_4 + \nonumber \\
  && \hspace{20.5mm}
  z_1 z_2 z_5 + z_2 z_3 z_5 + 
  z_3 z_4 z_5 + z_1 z_4 z_5 \big) \, .
\eeq
The two polynomials (and the corresponding Feynman integral) are symmetric 
under the re-labeling 
\beq
  (z_1, z_2) & \longleftrightarrow & (z_4, z_3) \, , \nonumber \\
  (z_1, z_3) & \longleftrightarrow & (z_2, z_4) \, ,
\eeq 
a property which reflects the symmetries of the graph, and which can be used 
to simplify the expressions resulting from the integration by parts identities 
in \eq{dw}. Indeed, it is a virtue of all approaches based on parameter space that
such symmetries under permutations of the graph propagators are manifest from
the beginning, and one does not have to deal with the degeneracy of possible graph 
parametrisations associated with different loop-momentum assignements, as is 
the case in momentum space.

To illustrate the use of these symmetries, consider the integral 
\beq
\label{mezza mela}
  I (1,1,1,1,1; - 1 + 3 \epsilon) \, = \, \int_{S_5} \eta_5 \, \,
  \frac{\mathcal{U}^{- 1 + 3 \epsilon}}{\mathcal{F}^{1+2\epsilon}}
\eeq
which is proportional to the Feynman integral associated with Fig.~\ref{fivedge}. 
\eq{U mezza mela} implies that 
\beq
  I(1,1,1,1,1; - 1 + 3 \epsilon) & = & I(2,1,2,1,1;- 2 + 3 \epsilon) + 
  I(2,1,2,1,1; - 2 + 3 \epsilon) \nonumber \\ 
  & + & I(2,1,2,1,1; - 2 + 3 \epsilon) + I(2,1,2,1,1; - 2 + 3 \epsilon) \nonumber \\
  & + & I(2,1,2,1,1; - 2 + 3 \epsilon) + I(2,1,2,1,2; - 2 + 3 \epsilon) \nonumber \\
  & + & I(2,1,2,1,2; - 2 + 3 \epsilon) + I(2,1,2,1,1; - 2 +  3 \epsilon) \, ,
\eeq
and the use of the symmetry properties of the graph polynomials reduces this 
equation to the much simpler form
\beq
\label{sum rule mm}
  I(1,1,1,1,1; - 1 + 3 \epsilon) & = & 2 I(2,1,2,1,1; - 2 + 3 \epsilon) + 
  2 I(2,1,1,2,1; - 2 + 3 \epsilon) \nonumber \\ 
  && \, + \, 4 I(2,1,1,1,2; - 2 + 3 \epsilon) \, .
\eeq
\eq{sum rule mm} is the first step necessary for reducing integral in \eq{mezza mela} 
to a linear combination of simpler integrals.

Consider now the integration by parts identity in \eq{d eta}, with $h = 1$, and with 
\beq
\label{defomega1}
  \omega_1 & = & - \, \eta_{\{2,3,4,5\}} \, \frac{z_3 \, 
  \mathcal{U}^{- 1 + 3 \epsilon}}{(1 + 2 \epsilon) \, \mathcal{F}^{1 + 2 \epsilon}} \, .
\eeq
One obtains then the integration by part identity
\beq
\label{IBP mm0}
  d \omega_1 \, = \, \frac{\eta_5}{(1 + 2 \epsilon) \, \mathcal{F}^{1 + 2 \epsilon}} \, 
  \frac{\partial}{\partial z_1} \left( z_3 \, \mathcal{U}^{- 1 + 3 \epsilon} \right) - 
  \frac{\eta_5}{\mathcal{F}^{2 + 2 \epsilon}} \left(z_3\, \mathcal{U}^{- 1 + 3 \epsilon}
  \right) \frac{\partial {\mathcal{F}}}{\partial z_1} \, .
\eeq 
Upon integration over the simplex $S_5$, this yields
\beq
\label{IBP mm1}
  \frac{\Omega_1}{1 + 2 \epsilon}  & = & - \frac{1 - 3 \epsilon}{1 + 2 \epsilon} \,
  \Big[ I(1,1,3,1,1; - 2 + 3 \epsilon) + I(1,1,2,2,1; - 2 + 3 \epsilon) + 
  I(1,1,2,1,2; - 2 + 3 \epsilon) \Big] \nonumber \\ 
  && - \, p^2\,  \Big[ I(1,2,3,1,1; - 1 + 3 \epsilon) + I(1,2,2,2,1; - 1 + 3 \epsilon) + 
  I(1,1,3,2,1; - 1 + 3 \epsilon) \nonumber \\ 
  && \hspace{1cm} + I(1,1,2,2,2; - 1 + 3 \epsilon) + I(1,2,2,1,2; - 1 + 3 \epsilon) \Big] \, .
\eeq
The integral $\Omega_1$ can be calculated by means of Stokes' theorem, with 
the result
\beq
  \Omega_1 \, \equiv (1 + 2 \epsilon) \int_{S_5} d \omega_1 \, = \, 
  (1 + 2 \epsilon) \int_{\partial S_5} \omega_1 \, = \, \int_{S_4} \eta_{\{2,3,4,5\}} \,
  \frac{z_3 \, \mathcal{U}^{- 1 + 3 \epsilon}}{\mathcal{F}^{1 + 2 \epsilon}} \, ,
\eeq
where the sign in the definition of $\omega_1$, \eq{defomega1}, is absorbed by 
the boundary $\partial S_5 = - S_{\{2,3,4,5\}}$, since the integrand vanishes on 
the other sub-simplexes comprising $\partial S_5$. The boundary term $\Omega_1$ 
is proportional to the Feynman integral obtained from the diagram in Fig.~\ref{fivedge} 
when the edge labelled 1 shrinks to a point, and the propagator corresponding to 
edge 3 is raised to the power of 2. This integral can be evaluated straightforwardly, 
yielding a product of Gamma functions (see for example Ref.~\cite{Chetyrkin:1981qh}).

A similar strategy can be applied to find the three other equations that are necessary 
to reduce the two-loop five-edge integral to simpler integrals. The resulting equations 
are 
\beq
\label{mm system 1}
  \frac{1}{1 + 2 \epsilon} \, \Omega_2 & = & - \, \frac{1 - 3 \epsilon}{1 + 2 \epsilon} \, 
  \Big[ I(1,1,3,1,1; - 2 + 3 \epsilon) + I(1,1,2,2,1; - 2 + 3 \epsilon)  \nonumber \\ 
  && \hspace{2cm} + \,
  I(1,2,2,1,1; - 2 + 3 \epsilon) + I(1,2,1,2,1; - 2 + 3 \epsilon) \Big] \nonumber \\ 
  && - \, p^2 \, \Big[ I(2,2,1,2,1; - 1 + 3 \epsilon) + I(1,1,3,2,1; - 1 + 3 \epsilon) \nonumber \\
  && \hspace{2cm} + \,
  I(1,2,3,1,1; - 1 + 3 \epsilon) + I(1,2,2,2,1; - 1 + 3  \epsilon) \Big] \, , \\
\label{mm system 2}
  0 & = & \frac{1}{1 + 2 \epsilon} \Big[ I(1,1,1,1,1; - 1 + 3 \epsilon) - 
  (1 - 3 \epsilon) \Big( I(1,2,2,1,1; - 2 + 3 \epsilon) \nonumber \\
  && \hspace{2cm} + \, I(1,2,1,2,1; - 2 + 3 \epsilon) + 
  I(1,1,2,1,2; - 2 + 3 \epsilon) \Big) \Big] \nonumber \\ 
  && \hspace{1cm} - \, p^2 \, \Big[ 2 I(1,2,2,2,1; - 1 + 3 \epsilon) + 
  I(2,2,1,2,1; - 1 + 3 \epsilon) \nonumber \\
  && \hspace{2cm} + \, I(1,2,2,1,2; - 1 + 3 \epsilon) + 
  I(1,1,2,2,2; - 1 + 3 \epsilon) \Big] \, , \\
\label{mm system 3}
  0 & = & \frac{1}{1 + 2 \epsilon} \Big[ I(1,1,1,1,1; - 1 + 3 \epsilon) - 4 (1 - 3 \epsilon)
  I(1,1,2,1,2; - 2 + 3 \epsilon) \Big] \nonumber \\ 
  && \hspace{1cm} - \, p^2 \, \Big[ 2 I(1,1,2,2,2; - 1 + 3 \epsilon) + 
  2 I(1,2,2,1,2; - 1 + 3 \epsilon) \Big] \, .
\eeq
In this case, the boundary term in \eq{mm system 1} is given by
\beq
\label{albound}
  \Omega_2 \, = \, (1 + 2 \epsilon) \int_{S_5} d \omega_2 \, = \, 
  (1 + 2 \epsilon) \int_{\partial S_5} \omega_2 \, = \, 
  \int_{S_4} \eta_{\{1,2,3,4\}} \, \frac{z_2 \, 
  \mathcal{U}^{- 1 + 3 \epsilon}}{\mathcal{F}^{1 + 2 \epsilon}} \, ,
\eeq 
corresponding to the diagram with the edge 5 shrunk to a point. Solving the 
system given by Eqs.~(\ref{mm system 1} - \ref{mm system 3}), together with 
\eq{IBP mm0} and \eq{IBP mm1} leads to the result
\beq
\label{resHA}
  \epsilon \, I(1,1,1,1,1; - 1 + 3 \epsilon) \, = \, - \, \Omega_1 + \Omega_2 \, ,
\eeq
which coincides with the well-known result of~\cite{Chetyrkin:1981qh}. As is 
the case with the momentum-space calculation, we note that in this case
one is actually not employing differential equations, since integration by parts
identities directly yield elementary integrals.


\section{Assessment and perspectives}
\label{AssPer}

In this paper, we have developed a projective framework to derive IBP identities 
and differential equations for Feynman integrals in parameter space, updating and 
extending ideas and results that first emerged half a century ago, prior to
modern developments. We have emphasised the significance of the early 
mathematical results reported in~\cite{Regge:1968rhi,Ponzano:1969tk,
Ponzano:1970ch,Regge:1972ns}, which resonate strikingly with contemporary 
research. These ideas from algebraic topology where turned into a concrete 
application to one-loop diagrams by Barucchi and Ponzano~\cite{Barucchi:1973zm,
Barucchi:1974bf}. In order to apply these results in the modern context, we 
have shown how the analysis extends naturally to dimensional regularisation,
we have generalised the results to the two-loop level (indeed we expect 
the technique to be applicable to all orders), and we have emphasised 
the role played by boundary terms in the IBP identities, noting that they
do not vanish in general, and in fact they provide a useful tool to link 
complicated integrals to simple ones. All these developments have explicitly 
been tested on relatively simple one- and two-loop diagrams, recovering known 
results, including the elliptic differential equation for the equal-mass sunrise diagram.

It is a natural question to ask how this method compares to the usual momentum-space 
approach. Clearly, this question cannot be answered in detail and in quantitative
computational terms at this stage, since this is just an exploratory study, while
momentum-space techniques have been honed through decades of optimisation.
We can however make a few observations already at this stage. 

First of all, it is clear that the parameter-space method offers, to say the least, a 
rather different organisation of the calculation of an integral family, as compared to 
momentum-space algorithms. This should be evident from the concrete cases 
examined in the text: for example, the integral basis arising naturally from the 
Barucchi-Ponzano theorem for the massless box is not the same as the conventional 
one, and the differential equations that emerge are different too~\cite{Henn:2014qga}. 

We note further that the way in which the lattice of different (integer) values of the 
indices $\nu_i$ is explored in parameter space appears different from standard IBPs. 
In the absence of boundary terms, parameter-space IBPs connect integrals with a fixed 
number of external legs, but different space-time dimensions. This is not necessarily 
a positive feature, since the goal of reduction algorithms is to a large extent to 
connect complicated integrals to simpler ones. It must however be noted that, in 
standard algorithms~\cite{Laporta:2000dsw}, the goal of achieving this simplification 
is reached in a rather roundabout way, through the ordering imposed in the recursive 
exploration of the index lattice. In parameter space, this simplifying step is specifically
associated with the novel feature of non-vanishing boundary terms, which give
lower-point integrals. These terms can in principle be reached in a simple way by 
suitably picking the initial values of the indices, as was done for the massless pentagon 
in \secn{OlMaP}. 

Continuing with the comparison, we observe that both the momentum-space 
algorithms and the projective one have a large degree of arbitrariness in their 
initialisation, which leaves room for optimisation. In the present case, there is clearly 
the possibility of many different choices for the functions $H_i (z)$ introduced in 
\eq{dw}. It is quite natural to choose the numerators of the original integral, as
we did, but it would be interesting to explore variations on this theme with an
eye to optimisation. On the other hand, in contrast to momentum-pace algorithms,
we observe that the parameter-space approach bypasses the ambiguity due
to the choice of loop-momentum routing, which can be non-negligible for complicated
diagrams; similarly, the issue of irreducible numerators is implicitly dealt with at
the momentum integration stage. These two aspects are among the consequences
of the fact that parameter space offers a minimal representation of Feynman 
integrals, transparently related to the symmetries of the original Feynman graph.

An especially promising aspect of the projective framework is its close connection
to the most significant algebraic structures associated with Feynman integrals. 
The Barucchi-Ponzano analysis can indeed be seen as an application of the results 
of Ref.~\cite{Ponzano:1970ch}, and it is notable that it succeeds not only in
constructing a system of differential equations for $n$-point one-loop integrals, 
but also in setting a bound on the size of the system, guaranteeing its closure, 
and providing an algorithmic construction. This is to be contrasted with the very 
large size of the systems of IBP identities that emerge in the intermediate stages 
of calculations in standard algorithms. It is clearly a goal of future research to 
extend these techniques and the analysis of Regge and collaborators to more 
complicated two- and higher-point integrals. In particular, studies on three-loop 
two-point functions and on two-loop three-point functions are currently ongoing, 
and steps towards the automation of the generation of IBPs in the projective 
framework are under way, with the goal of reaching state-of-the-art topologies 
such as two-loop penta- and hexa-boxes and three-loop four-point functions. When
complex multi-scale examples of this kind become available, a more thorough
comparison of the two approaches, including computational aspects, will
become possible.


\vspace{1cm}


\section*{Acknowledgments}

LM would like to thank the Regge Center for Algebra, Geometry and Theoretical 
Physics (then known as Arnold-Regge Center) for providing the opportunity
to discover Refs.~\cite{Regge:1968rhi,Ponzano:1969tk,Ponzano:1970ch,
Regge:1972ns}, as recounted in~\cite{DelDuca:2018nsu}. DA would like to 
thank Lina Alasfar, Dirk Kreimer, Till Martini, Jasper R. Nepveau, Maria C. 
Sevilla, Markus Schulze, Peter Uwer and Yingxuan Xu for the discussions 
and feedback during groups seminars and interships. We thank Simon
Badger, Stefan Weinzierl and Ben Page for useful discussions during the 
development of this project. Research supported in part by the Italian Ministry 
of University and Research (MIUR), under grant PRIN 20172LNEEZ. DA is 
funded by the Deutsche Forschungsgemeinschaft (DFG, German Research 
Foundation) - Projektnummer 417533893/GRK2575 "Rethinking Quantum 
Field Theory". 


\vspace{2cm}


\appendix


\section{IBPs for the one-loop massless box}
\label{App1}

Here we briefly present the system of linear equations necessary to close the 
differential equation system in \eq{system} for the one-loop massless box diagram. 
Our chosen basis integrals are $I(1,1,1,1; 2 \epsilon)$, $I(2,1,2,1; 2 \epsilon)$, 
$I(1,2,1,2; 2 \epsilon)$ and $I(2,2,2,2; 2 \epsilon)$: all the other relevant integrals
are determined by the linear system presented below. We find
\beq
\label{boxsystem}
  r \, I(3,1,3,1; 2 \epsilon) & = & \frac{1}{3 + \epsilon} \Big[ 2 I(2,1,2,1; 2 \epsilon) +
  2 \epsilon I(3,1,2,1; - 1 + 2 \epsilon) \Big] \, , \nonumber  \\
  I(2,2,2,2; 2 \epsilon) & = & \frac{1}{3 + \epsilon} \Big[ I(2,1,2,1; 2 \epsilon) + 
  2 \epsilon I(2,2,2,1; - 1 + 2 \epsilon) \Big] \, , \nonumber \\
  I(2,1,2,1; 2 \epsilon) & = & 2 I(2,2,2,1; - 1 + 2 \epsilon) + 2 I(3,1,2,1; - 1 + 2 \epsilon)  
  \, , \nonumber \\
  r \, I(3,2,3,2; 2 \epsilon) & = & \frac{1}{4 + \epsilon} \Big[ 2 I(2,2,2,2; 2 \epsilon) + 
  2 \epsilon I(3,2,2,2; - 1 + 2 \epsilon) \Big]  \, , \nonumber \\   
  I(2,2,2,2; 2 \epsilon) & = & 2 I(2,3,2,2; - 1 + 2 \epsilon) + 2 I(3,2,2,2; - 1 + 2 \epsilon) 
  \, , \nonumber \\
  I(2,3,2,3; 2 \epsilon) & = & \left( \frac{1}{2} + \epsilon \right) I(2,2,2,2; 2 \epsilon) +
  \frac{2 \epsilon I(2,3,2,2; - 1 + 2 \epsilon)}{4 + \epsilon} \, , \nonumber \\
  r \, I(3,2,2,2; - 1 + 2 \epsilon) & = & \frac{1}{3 + \epsilon} \Big[ I(2,2,1,2; - 1 + 2 \epsilon)
  - (1 - 2 \epsilon) I(2,2,2,2; - 2 + 2 \epsilon) \Big] \, , \nonumber \\
  I(2,3,2,2; - 1 + 2 \epsilon) & = & \frac{1}{3 + \epsilon} \Big[ I(2,2,2,1; - 1 + 2 \epsilon)
  - (1 - 2 \epsilon) I(2,2,2,2; - 2 + 2 \epsilon) \Big] \, , \nonumber \\
  r \, I(2,2,2,2; 2 \epsilon) & = & \frac{1}{3 + \epsilon} \Big[ I(1,2,1,2; 2 \epsilon) + 
  2 \epsilon I(2,2,1,2; - 1 + 2 \epsilon) \Big] \, .
\eeq
These are nine equations involving twelve independent integrals, to which one must
add the original integrals to be determined, $I(1,1,1,1; 2 \epsilon)$. The system is of 
course easily solved with elementary methods.


\section{Magnus exponentiation}
\label{App2}

In this Appendix, we will briefly review the Magnus exponentiation technique for 
solving systems of linear differential equations, and it application to the massless
box in \secn{OlMaB}. In general, one may consider a system of differential equations
of the form
\beq
\label{gensyst} 
  \partial_r \mathbf{b} (r) \, = \, M (r, \epsilon) \, \mathbf{b} (r) \, ,
\eeq
where $\mathbf{b} (r)$ is a vector of functions of $r$, and the matrix $M$ can 
be written as $M (r, \epsilon) \, = \, A (r) + \epsilon B(r)$.

In order to reduce the system to canonical form, consider a change of basis 
$\textbf{b}(r) \, = \, C(r) \textbf{b}' (r)$ where the matrix $C$ can depend on 
$r$ but not on $\epsilon$. The system for the vector $\textbf{b}' (r)$ is then 
determined by the matrix
\beq
\label{M'}
  M' (r, \epsilon) \, = \, C^{-1} (r) A (r) C(r) - C^{-1} (r) \partial_r C(r) 
  + \epsilon \, C^{-1} (r) B(r) C (r) \, .
\eeq
If one picks $C(r)$ such that $\partial_r C(r) \, = \, A(r) C(r)$, the system is 
reduced to canonical form. The general solution to this problem was reported 
in \cite{Magnus:1954zz}, and can be expressed by a formal expansion in $A(r)$, 
as
\beq
\label{solMagn}
  C (r) \, = \, \exp \left[ \int_{r_0}^r A(t) dt + \frac{1}{2} \int_{r_0}^r  dt_1\int^{t_1}_{r_0} 
  dt_2 \big[ A(t_1), A(t_2) \big] + \, \ldots \right] C_0 (r) \, .
\eeq
Since the goal is simply to eliminate the $\epsilon$-independent term, there 
is considerable freedom in choosing the base point $r_0$ and the matrix $C_0(r)$. 
In particular, the series reduces to a finite sum if the matrix $A(r)$ is upper triangular. 
In the specific case of massless box, \eq{system}, we can then proceed in steps. 
With a first change of basis, we make the $A$ matrix upper triangular. This is 
achieved with the rotation 
\beq
\label{Ctriang}
  C_{\rm tr} (r) \, = \, 
  \left(
  \begin{array}{cccc}
    \frac{1}{\epsilon^2 r} & 0 & 0 & 0 \\
    0 & \frac{1}{(2 + \epsilon) \epsilon^2 r^2} & 0 & 0 \\
    0 & \frac{1}{(2 + \epsilon) \epsilon^2 r} & - \frac{2}{(2 + \epsilon) \epsilon^2 r} & 0 \\
    0 & \frac{1}{(3 + \epsilon) (2 + \epsilon) \epsilon^2 r^2} & - \frac{1}{(3 + \epsilon) 
    (2 + \epsilon) \epsilon^2 r^2} &    \frac{1}{(3 + \epsilon) (2 + \epsilon) \epsilon^2 r^2} \\
  \end{array}
  \right) \, ,
\eeq
which reduces the system to
\beq
\label{triang sys}
  \partial_r \textbf{b}' (r) \, = \, \left[ \left(
  \begin{array}{cccc}
    \frac{1}{r} & - \frac{1}{r} & 0 & 0 \\
    0 & 0 & - \frac{1}{r} & \frac{1}{r} \\
    0 & 0 & 0 & \frac{1}{r} \\
    0 & 0 & 0 & \frac{1 - r}{r (1+ r)} \\
  \end{array}
  \right) 
  + \epsilon 
  \left(
  \begin{array}{cccc}
    0 & 0 & 0 & 0 \\
    0 & - \frac{1}{r} & 0 & 0 \\
    0 & - \frac{1}{2 r} & 0 & 0 \\
    0 & - \frac{1 - r}{2 r (1 + r)} & \frac{1}{r (1 + r)} & - \frac{1}{r (1 + r)} \\
  \end{array}
  \right) \right] \textbf{b}' (r) \, .
\eeq
The diagonal part $D(r)$ of the $\epsilon$-independent term is removed by the matrix 
\beq
\label{Cdiag}
  C_{\rm d} (r) \, = \, \exp \left[ \int_0^r D(t) dt \right] \, = \, 
  \left(
  \begin{array}{cccc}
    r & 0 & 0 & 0 \\
    0 & 1 & 0 & 0 \\
    0 & 0 & 1 & 0 \\
    0 & 0 & 0 & \frac{r}{(r+1)^2} \\
  \end{array}
  \right) \, ,
\eeq
which leads to the system
\beq
  \partial_r \textbf{b}'' (r) \, = \, \left[ \left(
  \begin{array}{cccc}
    0 & - \frac{1}{r^2} & 0 & 0 \\
    0 & 0 & - \frac{1}{r} & \frac{1}{(1 + r)^2} \\
    0 & 0 & 0 & \frac{1}{(1 + r)^2} \\
    0 & 0 & 0 & 0 \\
  \end{array}
  \right) 
  + \epsilon 
  \left(
  \begin{array}{cccc}
    0 & 0 & 0 & 0 \\
    0 & - \frac{1}{r} & 0 & 0 \\
    0 & - \frac{1}{2 r} & 0 & 0 \\
    0 & - \frac{(1 - r) (1 + r)}{2 r^2} & \frac{1 + r}{r^2} & - \frac{1}{r (1 + r)} \\
  \end{array}
  \right) \right] \textbf{b}'' (r) \, .
\eeq
One may now directly apply Magnus' theorem, with the final change of basis 
given by
\beq
\label{Cred}
  C_{\rm red} (r) \, = \, \left(
  \begin{array}{cccc}
    1 & \frac{1}{r} - 1 & \frac{r - \log r - 1}{r} & \frac{r - 2 \log [ (1+ r)/2 ] - 1}{2 r} \\
    0 & 1 & - \log r & \frac{r - 2 (r+1) \log [ (1+ r)/2 ] - 1}{2 (1 + r)} \\
    0 & 0 & 1 & \frac{r}{1 + r} \\
    0 & 0 & 0 & 1 \\
  \end{array}
  \right) \, .
\eeq
After this last step, the system is finally reduced to its canonical form, which can 
be solved iteratively. We write
\beq
\label{dF}
  \partial_r \textbf{b}''' (r) \, = \, \epsilon \, H (r) \textbf{b}''' (r) \, ,
\eeq
where the matrix $H$ is presented below, in Appendix~\ref{App3}.


\section{The matrix $H$ for the massless box in canonical form}
\label{App3}

The procedure discussed in Appendix \ref{App2} leads to a matrix containing at most
logarithms of the kinematical variable $s/t$. The matrix elements are given below.
\small
$$
\begin{aligned}
  & H_{11} (r) \, = \, 0 \, , \\
  & H_{12} (r) \, = \, \frac{1}{4} \frac{r^2 + 3}{r^2} \, , \\
  & H_{13} (r) \, = \, - \frac{1}{4} \frac{r^2 \ln r + 3 \ln r - 2 r + 2}{r^2} \, , \\
  & H_{14} (r) \, = \, \frac{1}{8} \frac{1}{r^2 (r+1)^2} \left[ r^4 \ln r - 
  2 r^4 \ln \frac{r + 1}{2} + 2 r^3 \ln r - 4 r^3 \ln \frac{r + 1}{2} + r^4 + 
  4 r^2 \ln r - 8 r^2 \ln \frac{r + 1}{2} + 2 r^{3} \right. \\
  & \hspace{3.5cm} \left. + \, 6 r \ln r -12 r \ln \frac{r + 1}{2} - 4 r^2 + 3 \ln r - 
  6 \ln \frac{r + 1}{2} + 2 r - 1 \right] \, , \\
  & H_{21} (r) \, = \, 0 \, , \\
  & H_{22} (r) \, = \, - \frac{1}{4 r^2} \left[ 2  r^2 \ln r - 2 r^2 \ln \frac{r + 1}{2} +
  r^2 + 2 \ln \frac{r + 1}{2} + 2 r + 1 \right] \\
  & H_{23} (r) \, = \, \frac{1}{4 r^2} \left[ 2 r^2 \ln^2 r - 2 r^2 \ln r \ln \frac{r + 1}{2} +
   r^2 \ln r + 2 \ln r \ln \frac{r + 1}{2} - 2 r \ln r + 4 r \ln \frac{r + 1}{2} + \ln r \right. \\
  & \hspace{2.4cm} \left. + \, 4 \ln \frac{r + 1}{2} - 2 r + 2 \right] \, , \\
  & H_{24} (r) \, = \, \frac{1}{2 (r + 1)} \left( r \ln r - 2 r \ln \frac{r + 1}{2} + \ln r - 
  2 \ln \frac{r + 1}{2} + r - 1 \right) \left[ - \frac{1}{r} - \frac{\ln r}{2 r} - 
  \frac{r - 1}{4 r^2} \bigg( 2 r \ln r \right. \\
  & \hspace{1.5cm} \left. - \, 2 r \ln \frac{r + 1}{2} - 2 \ln \frac{r + 1}{2} + r - 1 
  \bigg) \right] - \frac{(r - 1)}{4 r^2 (r + 1)} \left( 2 r \ln r - 2 r \ln \frac{r + 1}{2} - 
  2 \ln \frac{r + 1}{2} + r - 1 \right) \\
  & \hspace{1.5cm}  + \, \frac{1}{2 r (r + 1)^2} \left( 2 r \ln r - 2 r \ln \frac{r + 1}{2} - 
  2 \ln \frac{r + 1}{2} + r - 1 \right) \, , \\
  & H_{31} (r) \, = \, 0 \, , \\
  & H_{32} (r) \, = \, - \frac{r^2+1}{4 r^2} \, , \\
  & H_{33} (r) \, = \, \frac{r^2 \ln r + \ln r - 2 r + 2}{4 r^2} \, , \\
  & H_{34} (r) \, = \, - \frac{1}{8 r^2 (r+1)^2} \left[ r^4 \ln r - 2 r^4 \ln \frac{r + 1}{2} +
  2 r^3 \ln r - 4 r^3 \ln \frac{r + 1}{2} + r^4 + 2 r^2 \ln r - 4 r^2 \ln \frac{r + 1}{2} \right. \\
  & \hspace{3.8cm} \left. + 2 r^{3} + 2 r \ln r - 4 r \ln \frac{r + 1}{2} - 6 r^2 + \ln r - 
  2 \ln \frac{r + 1}{2} + 2 r + 1 \right] \\
  & H_{41} (r) \, = \, 0 \, , \\
  & H_{42} (r) \, = \, \frac{r^2 - 1}{2 r^2} \, , \\
  & H_{43} (r) \, = \, - \frac{r^2 \ln r - \ln r - 2 r - 2}{2 r^2} \, , \\
  & H_{44} (r) \, = \, \frac{r - 1}{4 r^2} \left[ r \ln r - 2 r \ln \frac{r + 1}{2} + \ln r - 
  2 \ln \frac{r + 1}{2} + r + 1 \right] - \frac{1}{r (r+1)} \, .
\end{aligned}
$$
Once the matrix $H$ is known, the solution fo the differential equation for the 
massless box can be determined by iteration in $\epsilon$ by standard methods.


\section{IBPs for the two-loop sunrise integral}
\label{App4}

Here we present the linear system necessary to close the system of differential 
equations in \eq{diffSUN}. Our chosen basis integrals are $I(111,3\epsilon)$, 
$I(211,1+3\epsilon)$, and the boundary contribution  $B$. All other relevant
integrals are determined by the following set of IBP equations.
\beq
\label{surisyst}
  I(1,1,1; 3 \epsilon) - r I(2,2,2; 3 \epsilon) + 3 I(2,1,1; 1 + 3 \epsilon) 
  & = & \, 0 \, , \qquad \nonumber \\
  I(2,2,2; 3 \epsilon) + 2 I(3,2,1; 3 \epsilon) - I(2,1,1; 1 + 3 \epsilon) 
  & = & 0 \, , \qquad \nonumber \\
  \frac{1}{2 + 2 \epsilon} \Big[ (1 + 3 \epsilon) \big( I(2,2,2; 3 \epsilon) +  
  I(3,2,1; 3 \epsilon) \big) + I(2,1,1; 1 + 3 \epsilon) \Big] \hspace{4cm}
  && \qquad \nonumber \\ + \,  
  (1 - r) I(3,2,2; 1 + 3 \epsilon) + I(3,2,2; 1 + 3 \epsilon) + I(4,2,1; 1 + 3 \epsilon) 
  + 2 I(2,2,1; 2 + 3 \epsilon) & = & 0 \, , \qquad \nonumber \\
  - 2 I(3,2,2; 1 + 3 \epsilon) - I(3,3,1; 1 + 3 \epsilon) + I(2,2,1; 2 + 3 \epsilon) 
  & = & 0 \, , \qquad \nonumber \\
  \frac{1 + 3 \epsilon}{1 + \epsilon} I(3,2,1; 3 \epsilon) + 
  (1 - r) I(3,3,1; 1 + 3 \epsilon) + 2 I(4,2,1; 1 + 3 \epsilon) + 2 I(2,2,1; 2 + 3 \epsilon) 
  & = & B \, , \qquad \nonumber \\
  \frac{1}{1 + \epsilon} \Big[ (1 + 3 \epsilon) I(3,2,1; 3 \epsilon) + 
  I(2,1,1; 1 + 3 \epsilon) \Big]  \hspace{4cm}
  && \qquad \\ 
  + \, (1 - r) I(3,2,2; 1 + 3 \epsilon) + 
  2 I(3,3,1; 1 + 3 \epsilon) + 2  I(3,1,1; 2 + 3 \epsilon) & = & 0 \, , \qquad \nonumber \\
 - I(3,2,2; 1 + 3 \epsilon) - 2 I(4,2,1; 1 + 3 \epsilon) +  
 I(3,1,1; 2 + 3 \epsilon) & = & 0 \, . \qquad \nonumber
\eeq
These are seven equations involving nine independent integrals, two of which are
the chosen basis integrals. The system is of course easily solved with elementary 
methods.


\vspace{1cm}



\end{document}